\documentclass[copyright,creativecommons]{eptcs}

\usepackage[utf8]{inputenc}
\usepackage[T1]{fontenc}
\usepackage{amsmath,amssymb,amsfonts,amsthm}
\usepackage{mathtools}
\usepackage{graphicx}
\usepackage{xcolor}
\usepackage{tikz}
\usepackage{pgfplots}
\usepackage{tikz-cd}
\usepackage{braket}
\usepackage{array}
\usepackage{booktabs}
\usepackage{enumitem}
\usepackage{longtable}
\usepackage{multirow}
\usepackage{nicefrac}
\usepackage{needspace}
\usepackage{stmaryrd}
\usepackage{xspace}
\pdfmapfile{+stmaryrd.map}
\pdfmapfile{+rsfs.map}
\hypersetup{
  colorlinks=true,
  linkcolor=blue,
  citecolor=blue,
  urlcolor=blue,
  pdftitle={Minimality of the Pure Qubit ZX Calculus},
  pdfauthor={Harry K. Stoltz and Renaud Vilmart}
}
\usepackage[nameinlink]{cleveref}

\usetikzlibrary{trees}
\usetikzlibrary{topaths}
\usetikzlibrary{decorations.pathmorphing}
\usetikzlibrary{decorations.markings}
\usetikzlibrary{matrix,backgrounds,folding}
\usetikzlibrary{chains,scopes,positioning,fit}
\usetikzlibrary{arrows,shadows}
\usetikzlibrary{calc}
\usetikzlibrary{shapes,shapes.geometric,shapes.misc}

\pgfdeclarelayer{edgelayer}
\pgfdeclarelayer{nodelayer}
\pgfsetlayers{background,edgelayer,nodelayer,main}
\pgfplotsset{compat=1.18}
\tikzstyle{none}=[inner sep=0mm]
\tikzstyle{every loop}=[]
\tikzstyle{env}=[copoint,regular polygon rotate=0,minimum width=0.2cm, fill=black]

\tikzstyle{hadamard edge}=[-,color=blue,dashed,dash pattern=on 2pt off 0.7pt]
\tikzstyle{black edge}=[-,color=black,dashed,dash pattern=on 2pt off 0.7pt]

\tikzstyle{every picture}=[baseline=-0.25em]
\tikzstyle{dotpic}=[scale=0.5]
\tikzstyle{diredges}=[every to/.style={diredge}]
\tikzstyle{dot graph}=[shorten <=-0.1mm,shorten >=-0.1mm,scale=0.6]
\tikzstyle{plot point}=[circle,fill=black,minimum width=2mm,inner sep=0]

\tikzstyle{braceedge}=[decorate,decoration={brace,amplitude=2mm,raise=-1mm}]
\tikzstyle{small braceedge}=[decorate,decoration={brace,amplitude=1mm,raise=-1mm}]
\tikzstyle{left hook arrow}=[left hook-latex]
\tikzstyle{right hook arrow}=[right hook-latex]

\tikzstyle{black dot}=[inner sep=0.7mm,minimum width=0pt,minimum height=0pt,fill=black,draw=black,shape=circle]

\tikzstyle{dot}=[black dot]
\tikzstyle{smalldot}=[inner sep=0.4mm,minimum width=0pt,minimum height=0pt,fill=black,draw=black,shape=circle]
\tikzstyle{white dot}=[dot,fill=white]
\tikzstyle{antipode}=[white dot,inner sep=0.3mm,font=\footnotesize]
\tikzstyle{smallwhitedot}=[smalldot,fill=white]
\tikzstyle{alt white dot}=[white dot,label={[xshift=3.07mm,yshift=-0.05mm,font=\footnotesize]left:$*$}]
\tikzstyle{gray dot}=[dot,fill=gray!40!white]
\tikzstyle{smallgraydot}=[smalldot,fill=gray!40!white]
\tikzstyle{box vertex}=[draw=black,rectangle]
\tikzstyle{small box}=[box vertex,fill=white]
\tikzstyle{whitebg}=[fill=white,inner sep=2pt]
\tikzstyle{graph state vertex}=[sg vertex,fill=black]

\tikzstyle{wide copoint}=[fill=white,draw=black,shape=isosceles triangle,shape border rotate=90,isosceles triangle stretches=true,inner sep=1pt,minimum width=1.5cm,minimum height=5mm]
\tikzstyle{wide point}=[fill=white,draw=black,shape=isosceles triangle,shape border rotate=-90,isosceles triangle stretches=true,inner sep=1pt,minimum width=1.5cm,minimum height=4mm]
\tikzstyle{very wide copoint}=[fill=white,draw=black,shape=isosceles triangle,shape border rotate=-90,isosceles triangle stretches=true,inner sep=1pt,minimum width=2.5cm,minimum height=4mm]
\tikzstyle{very wide empty copoint}=[draw=black,shape=isosceles triangle,shape border rotate=-90,isosceles triangle stretches=true,inner sep=1pt,minimum width=2.5cm,minimum height=4mm]
\tikzstyle{symm}=[ultra thick,shorten <=-1mm,shorten >=-1mm]

\tikzstyle{square box}=[rectangle,fill=white,draw=black,minimum height=5mm,minimum width=5mm,font=\small]
\tikzstyle{square gray box}=[rectangle,fill=gray!30,draw=black,minimum height=6mm,minimum width=6mm]
\tikzstyle{wide medium box}=[shape=rectangle,text height=1.5ex,text depth=0.25ex,yshift=0.5mm,fill=white,draw=black,minimum height=10mm,yshift=-0.5mm,minimum width=7.5mm,font={\small}]
\tikzstyle{pointer edge}=[->,very thick,gray]
\tikzstyle{copoint}=[regular polygon,regular polygon sides=3,draw=black,scale=0.75,inner sep=-0.5pt,minimum width=7mm,fill=white]
\tikzstyle{point}=[regular polygon,regular polygon sides=3,draw=black,scale=0.75,inner sep=-0.5pt,minimum width=7mm,fill=white,regular polygon rotate=180]
\tikzstyle{gray point}=[point,fill=gray!40!white]
\tikzstyle{gray copoint}=[copoint,fill=gray!40!white]

\newcommand{\edgearrow}{{\arrow[black]{>}}}
\newcommand{\edgetick}{{\arrow[black,scale=0.7,very thick]{|}}}

\tikzstyle{diredge}=[->]
\tikzstyle{rdiredge}=[<-]
\tikzstyle{medium diredge}=[->]

\tikzstyle{short diredge}=[->]
\tikzstyle{halfedge}=[-)]
\tikzstyle{other halfedge}=[(-]
\tikzstyle{freeedge}=[(-)]
\tikzstyle{white edge}=[line width=5pt,white]
\tikzstyle{tick}=[postaction=decorate,decoration={markings, mark=at position 0.5 with \edgetick}]
\tikzstyle{small map edge}=[|-latex, gray!60!blue, shorten <=0.9mm, shorten >=0.5mm]
\tikzstyle{thick dashed edge}=[very thick,dashed,gray!40]
\tikzstyle{map edge}=[|-latex,very thick, gray!40, shorten <=1mm, shorten >=0.5mm]
\tikzstyle{tickedge}=[postaction=decorate,
  decoration={markings, mark=at position 0.5 with \edgetick}]
\tikzstyle{dirtickedge}=[postaction=decorate,
  decoration={markings, mark=at position 0.5 with \edgetick},
  decoration={markings, mark=at position 0.85 with \edgearrow}]
\tikzstyle{dirdoubletickedge}=[postaction=decorate,
  decoration={markings, mark=at position 0.4 with \edgetick},
  decoration={markings, mark=at position 0.6 with \edgetick},
  decoration={markings, mark=at position 0.85 with \edgearrow}]

\makeatletter
\newcommand{\boxshape}[3]{%
\pgfdeclareshape{#1}{
\inheritsavedanchors[from=rectangle] 
\inheritanchorborder[from=rectangle]
\inheritanchor[from=rectangle]{center}
\inheritanchor[from=rectangle]{north}
\inheritanchor[from=rectangle]{south}
\inheritanchor[from=rectangle]{west}
\inheritanchor[from=rectangle]{east}
\backgroundpath{
\southwest \pgf@xa=\pgf@x \pgf@ya=\pgf@y
\northeast \pgf@xb=\pgf@x \pgf@yb=\pgf@y

\@tempdima=#2
\@tempdimb=#3

\pgfpathmoveto{\pgfpoint{\pgf@xa - 5pt + \@tempdima}{\pgf@ya}}
\pgfpathlineto{\pgfpoint{\pgf@xa - 5pt - \@tempdima}{\pgf@yb}}
\pgfpathlineto{\pgfpoint{\pgf@xb + 5pt + \@tempdimb}{\pgf@yb}}
\pgfpathlineto{\pgfpoint{\pgf@xb + 5pt - \@tempdimb}{\pgf@ya}}
\pgfpathlineto{\pgfpoint{\pgf@xa - 5pt + \@tempdima}{\pgf@ya}}
\pgfpathclose
}
}}

\boxshape{NEbox}{0pt}{8pt}
\boxshape{SEbox}{0pt}{-8pt}
\boxshape{NWbox}{8pt}{0pt}
\boxshape{SWbox}{-8pt}{0pt}
\makeatother

\tikzstyle{map}=[draw,shape=NEbox,inner sep=7pt]
\tikzstyle{mapdag}=[draw,shape=SEbox,inner sep=7pt]
\tikzstyle{maptrans}=[draw,shape=SWbox,inner sep=7pt]
\tikzstyle{mapconj}=[draw,shape=NWbox,inner sep=7pt]

\tikzstyle{probs}=[shape=semicircle,fill=gray!40!white,draw=black,shape border rotate=180,minimum width=1.2cm]

\tikzstyle{arrs}=[-latex,font=\small,auto]
\tikzstyle{arrow plain}=[arrs]
\tikzstyle{arrow dashed}=[dashed,arrs]
\tikzstyle{arrow bold}=[very thick,arrs]
\tikzstyle{arrow hide}=[draw=white!0,-]
\tikzstyle{arrow reverse}=[latex-]
\tikzstyle{cdnode}=[]

\tikzstyle{gn}=[dot,fill=lime!50,minimum width=0.2cm,inner sep=0.5pt,font=\footnotesize]
\tikzstyle{rn}=[dot,fill=red!50,inner sep=0.5pt,minimum width=0.2cm,font=\footnotesize]
\tikzstyle{bn}=[dot,fill=blue,minimum width=0.3cm]

\tikzstyle{rc}=[dot,thick,fill=white,draw = red,minimum width=0.2cm,inner sep=0.5pt,font=\footnotesize]
\tikzstyle{gc}=[dot,thick,fill=white,draw= lime,inner sep=0.5pt,minimum width=0.2cm,font=\footnotesize]
\tikzstyle{bc}=[dot,thick,fill=white,draw= blue,minimum width=0.3cm]

\tikzstyle{label}=[circle,fill=white,minimum width=0.3cm]

\tikzstyle{H box}=[rectangle,fill=yellow,draw=black,xscale=1,yscale=1,font=\small,inner sep=0.75pt,minimum width=0.15cm,minimum height=0.15cm]

\tikzstyle{clocklabel}=[dot,fill=yellow,draw=black,font=\tiny,inner sep=0.75pt]

\tikzstyle{rsn}=[circle split,draw,fill=red,font=\tiny,inner sep=0.75pt]
\tikzstyle{gsn}=[circle split,draw,fill=lime,font=\tiny,inner sep=0.75pt]
\tikzstyle{bsn}=[circle split,draw,fill=blue,font=\tiny,inner sep=0.75pt]

\tikzstyle{rsc}=[circle split,thick,draw= red,draw,fill=white,font=\tiny,inner sep=0.75pt]
\tikzstyle{gsc}=[circle split,thick,draw= lime,draw,fill=white,font=\tiny,inner sep=0.75pt]
\tikzstyle{bsc}=[circle split,thick,draw= blue,draw,fill=white,font=\tiny,inner sep=0.75pt]

\tikzstyle{cnot}=[fill=white,shape=circle,inner sep=-1.4pt]
\tikzstyle{wire label}=[font=\tiny, auto]

\tikzset{
	gn/.append style={fill={rgb,255:red,216;green,248;blue,216}},
	rn/.append style={fill={rgb,255:red,232;green,165;blue,165}}
}

\newcommand{\vilmartboldfig}[1]{\begingroup\bfseries\let\textnormal\textbf\input{Vilmart/Arxiv/arXiv-1812.09114v1/figures/#1.tikz}\endgroup}
\newcommand{\vilmartlabeledfig}[2]{%
	\begin{tikzpicture}[baseline=(fig.base)]
		\node [inner sep=0pt] (fig) at (0,0) {\input{Vilmart/Arxiv/arXiv-1812.09114v1/figures/#1.tikz}};
		\node [style=none] at ([xshift=0.25cm,yshift=-0.32cm]fig.center) {#2};
	\end{tikzpicture}%
}

\newcommand{\typedfig}[1]{\IfFileExists{latex_typed/b_rule_note/figures/#1.tikz}{\input{latex_typed/b_rule_note/figures/#1.tikz}}{\input{latex_typed_b/b_rule_note/figures/#1.tikz}}}
\tikzstyle{tikzfig}=[baseline=-0.25em,scale=0.5]

\newcommand{\widetikzitinput}[1]{%
	\begin{center}
	\makebox[\linewidth][c]{%
	\begingroup
	\scalebox{0.8}{{\tikzstyle{every picture}=[tikzfig]\input{#1}}}%
	\endgroup
}%
	\end{center}
}
\newcommand{\hkstikzitstyles}{%
	\tikzstyle{Z dot}=[inner sep=0mm,minimum size=2mm,shape=circle,draw=black,fill={rgb,255:red,216;green,248;blue,216}]%
	\tikzstyle{X dot}=[Z dot,fill={rgb,255:red,232;green,165;blue,165}]%
	\tikzstyle{Z phase dot}=[minimum size=5mm,font={\footnotesize\boldmath},shape=rectangle,rounded corners=2mm,inner sep=0.2mm,outer sep=-2mm,scale=0.8,draw=black,fill={rgb,255:red,216;green,248;blue,216}]%
	\tikzstyle{X phase dot}=[Z phase dot,fill={rgb,255:red,232;green,165;blue,165},font={\footnotesize\boldmath}]%
	\tikzstyle{hadamard}=[fill=yellow,draw=black,shape=rectangle,inner sep=0.6mm,minimum height=1.5mm,minimum width=1.5mm]%
	\tikzstyle{empty diagram}=[draw={gray!40!white},dashed,shape=rectangle,minimum width=1cm,minimum height=1cm]%
}
\newcommand{\hkstikzitinput}[1]{%
	\begingroup
	\hkstikzitstyles
	\begingroup
	\scalebox{0.8}{{\tikzstyle{every picture}=[tikzfig]\input{#1}}}%
	\endgroup
	\endgroup
}
\newcommand{\primeirtikzinput}[1]{%
	\begingroup
	\def\dref##1##2{\hyperref[eq:prime-ir-##1]{\textcolor{blue}{\ensuremath{(##2)}}}}%
	\def\dlink##1##2{%
		\ifnum\pdfstrcmp{##1}{Fig3H0}=0
			\hyperref[eq:ir-derivation-L1]{\textcolor{blue}{\ensuremath{(##2)}}}%
		\else
			\hyperref[eq:ir-derivation-L2]{\textcolor{blue}{\ensuremath{(##2)}}}%
		\fi
	}%
	\hkstikzitinput{HKS/figures/ir-derivation/zxoptprime-73/#1}%
	\endgroup
}
\newcommand{\widehkstikzitinput}[1]{%
	\begin{center}
	\makebox[\linewidth][c]{\hkstikzitinput{#1}}%
	\end{center}
}

\newcommand{\lemmalink}[2]{\hyperref[eq:ir-derivation-L#1]{(\ref*{eq:ir-derivation-L#1})}}
\newcommand{\rulelink}[2]{#2}

\newcommand{\tombstone}{\par\nobreak\hfill\ensuremath{\blacksquare}}
\newenvironment{irproof}
	{\begingroup\begin{proof}[\normalfont\bfseries Proof]}
	{\end{proof}\endgroup}
\newcommand{\zxfig}[1]{\vcenter{\hbox{\scalebox{1}{\input{latex_typed/b_rule_note/figures/#1.tikz}}}}}
\newcommand{\zxbig}[1]{\vcenter{\hbox{\scalebox{1}{\input{latex_typed/b_rule_note/figures/#1.tikz}}}}}
\newcommand{\zxsmall}[1]{\vcenter{\hbox{\scalebox{0.85}{\input{latex_typed/b_rule_note/figures/#1.tikz}}}}}

\newcommand{\zxopt}{\ensuremath{\textup{ZX}_{\textsf{opt}}}\xspace}
\newcommand{\zxv}{\ensuremath{\textup{ZX}_{\textsf{V}}}\xspace}
\newcommand{\zxvprime}{\ensuremath{\textup{ZX}_{\textsf{V}}'}\xspace}
\newcommand{\zxoptprime}{\ensuremath{\textup{ZX}_{\textsf{opt}}'}\xspace}
\newcommand{\intf}[1]{\ensuremath{\left\llbracket #1 \right\rrbracket}}
\newcommand{\sem}[1]{\left\llbracket #1 \right\rrbracket}
\newcommand{\bool}{\mathbb{B}}

\newcommand{\ocross}{\otimes}

\newcommand{\sOneTableDiagram}{\vilmartlabeledfig{spider-1}{(S)}}

\newcommand{\sLhsDiagram}{\vcenter{\hbox{%
\begin{tikzpicture}[font={\footnotesize}]
	\begin{pgfonlayer}{nodelayer}
		\node [style=none] (0) at (-1, -0) {\rotatebox[origin=c]{63.43}{$~\cdots~$}};
		\node [style=gn, minimum width=0.5 cm] (3) at (-0.5, -0.25) {$\beta$};
		\node [style=none] (5) at (-0.75, -0.75) {};
		\node [style=none] (6) at (-0.25, -0.75) {};
		\node [style=none] (8) at (-1.75, -0.5) {};
		\node [style=none] (9) at (-1.25, -0.5) {};
		\node [style=none] (12) at (-1.75, 0.75) {};
		\node [style=none] (13) at (-0.25, 0.5) {};
		\node [style=none] (14) at (-1.25, 0.75) {};
		\node [style=gn, minimum width=0.5 cm] (16) at (-1.5, 0.25) {$\alpha$};
		\node [style=none] (17) at (-0.75, 0.5) {};
		\node [style=none, yshift=-3.5pt] (18) at (-1.5, 0.75) {$~\cdots~$};
		\node [style=none, yshift=2.5pt] (20) at (-0.5, -0.75) {$~\cdots~$};
		\node [style=none, yshift=2.5pt] (21) at (-1.5, -0.5) {$~\cdots~$};
		\node [style=none, yshift=-3.5pt] (22) at (-0.5, 0.5) {$~\cdots~$};
	\end{pgfonlayer}
	\begin{pgfonlayer}{edgelayer}
		\draw (3) to (13.center);
		\draw (3) to (5.center);
		\draw (3) to (6.center);
		\draw (16) to (8.center);
		\draw (16) to (9.center);
		\draw [bend right, looseness=1.00] (16) to (3);
		\draw [bend left, looseness=1.00] (16) to (3);
		\draw (16) to (12.center);
		\draw (16) to (14.center);
		\draw (3) to (17.center);
	\end{pgfonlayer}
\end{tikzpicture}}}}
\newcommand{\sRhsDiagram}{\vcenter{\hbox{%
\begin{tikzpicture}[font={\footnotesize}]
	\begin{pgfonlayer}{nodelayer}
		\node [style=gn] (z) at (0, 0) {$\alpha+\beta$};
		\node [style=none] (t1) at (-0.5, 0.75) {};
		\node [style=none] (t2) at (0.5, 0.75) {};
		\node [style=none] (td) at (0, 0.5) {$~\cdots~$};
		\node [style=none] (b1) at (-0.5, -0.75) {};
		\node [style=none] (b2) at (0.5, -0.75) {};
		\node [style=none] (bd) at (0, -0.75) {$~\cdots~$};
	\end{pgfonlayer}
	\begin{pgfonlayer}{edgelayer}
		\draw (t1.center) to (z);
		\draw (t2.center) to (z);
		\draw (z) to (b1.center);
		\draw (z) to (b2.center);
	\end{pgfonlayer}
\end{tikzpicture}}}}

\newcommand{\hAlphaRhsDiagram}{\vcenter{\hbox{%
\begin{tikzpicture}
	\begin{pgfonlayer}{nodelayer}
		\node [style=gn] (z) at (0, 0) {\footnotesize$~\alpha~$};
		\node [style=none] (t1) at (-0.5, 0.75) {};
		\node [style=none] (t2) at (0.5, 0.75) {};
		\node [style=none] (b1) at (-0.5, -0.75) {};
		\node [style=none] (b2) at (0.5, -0.75) {};
		\node [style=none] (td) at (0, 0.75) {$\cdots$};
		\node [style=none] (bd) at (0, -0.75) {$\cdots$};
	\end{pgfonlayer}
	\begin{pgfonlayer}{edgelayer}
		\draw[bend left=23] (z) to (t1.center);
		\draw[bend right=23] (z) to (t2.center);
		\draw[bend right=23] (z) to (b1.center);
		\draw[bend left=23] (z) to (b2.center);
	\end{pgfonlayer}
\end{tikzpicture}}}}

\newcommand{\euRhsNoScalarsDiagram}{\vcenter{\hbox{\scalebox{0.85}{%
\begin{tikzpicture}
	\begin{pgfonlayer}{nodelayer}
		\node [style=rn] (0) at (0, 0.55) {$\beta_1$};
		\node [style=gn] (1) at (0, 0) {$\beta_2$};
		\node [style=rn] (2) at (0, -0.55) {$\beta_3$};
		\node [style=none] (3) at (0, 1) {};
		\node [style=none] (4) at (0, -1) {};
	\end{pgfonlayer}
	\begin{pgfonlayer}{edgelayer}
		\draw [style=none] (3.center) to (4.center);
	\end{pgfonlayer}
\end{tikzpicture}}}}}

\newtheorem{theorem}{Theorem}[section]
\newtheorem{lemma}[theorem]{Lemma}

\theoremstyle{definition}
\newtheorem{definition}[theorem]{Definition}

\theoremstyle{remark}

\title{Minimality of the Pure Qubit ZX Calculus}
\date{}
\author{
  Harry K. Stoltz$\null^{1}$ \and
  Renaud Vilmart$\null^{2}$ \and
  \institute{$\null^{1}$Courant Institute of Mathematical Sciences, New York University}
  \institute{$\null^{2}$Universit\'e de Lorraine, CNRS, Inria, LORIA, F 54000 Nancy, France}
}
\def\titlerunning{Minimality of the Pure Qubit ZX Calculus}
\def\authorrunning{H.K. Stoltz \& R. Vilmart}
\newcommand{\publicationstatus}{}

\begin{document}

\maketitle

\begin{abstract}
The ZX calculus is a graphical language for reasoning about quantum processes. In this paper, we develop a minimal pure-qubit ZX calculus based on the work of Vilmart~\cite{8785765}, Backens, Perdrix, and Wang~\cite{Backens_2020}, and Stoltz~\cite{stoltz2026simplified}. This resolves a problem that has remained open for nearly a decade, since completeness was first proved. Specifically, we show that \((I_r)\) is derivable and establish the necessity of \((B)\) and \((I_g)\), yielding two complete and minimal rulesets.
\end{abstract}

\section{Introduction}
\label{sec:introduction}

Introduced by Bob Coecke and Ross Duncan in 2008~\cite{CoeckeDuncan2008,CoeckeDuncan2011}, the ZX calculus is a graphical language for reasoning about quantum processes. It has found applications in quantum-circuit optimisation and extraction, measurement-based quantum computation, and surface-code lattice surgery~\cite{DuncanEtAl2020,DuncanPerdrix2010,deBeaudrapHorsman2020}. The first completeness results for the full pure-qubit ZX calculus appeared in 2017 and 2018~\cite{NgWang2017Universal,HadzihasanovicNgWang2018,JeandelPerdrixVilmart2018Beyond}.

In 2019, Vilmart gave two near-optimal axiomatisations, reducing the fifteen displayed equations of the previous presentation to complete rulesets with nine and eight named rules, respectively, and revealing their fundamental properties~\cite{8785765}. In 2020, Miriam Backens, Simon Perdrix, and Quanlong Wang crafted another near-minimal ruleset for the stabilizer ZX calculus~\cite{Backens_2020}. Their presentation consists of nine explicit rewrite rules, together with the meta-rule ``only connectivity matters,'' and was refined over a series of papers beginning with their earlier work~\cite{Backens_2017}. Interestingly, both rulesets (stabilizer and pure-qubit) remained unresolved at roughly the same points: proving that the bialgebra law is necessary, and the status of the red and green compact structures.

As for the bialgebra, this is surprising since it is one of the cornerstones of the ZX calculus, since it encodes a property known as \emph{strong complementarity}~\cite{CoeckeDuncan2011}. As for the red and green compact-structure coincidence, it is not obvious whether inclusion of one of them implies the other in the ruleset, but it is clear that at least one of them must be necessary.

The stabilizer case was finally solved in 2026~\cite{stoltz2026simplified}. As it turned out, the red and green compact-structure rules were both necessary in the stabilizer calculus, by traditional arguments. However, the proof of the necessity of the bialgebra law required passing through a relatively nonstandard PROP built from modules over a finite ring with nilpotent elements. This construction allowed for all the other axioms to hold, but broke the strongly complementary nature of the red and green spiders, hence proving that the ruleset posited in~\cite{Backens_2020} was indeed minimal.

In this paper, we resolve the analogous problem for the pure-qubit ZX calculus. In this fragment, phases are less restricted, and the troublesome \((EU)\) rule allows for a wide family of configurations. It is not obvious that certain rules truly can be separated from each other. Indeed, unlike the stabilizer case, it is actually possible to derive the red identity rule from the green one! Predictably, the bialgebra law is indeed necessary in this ruleset, and the proof relies on an elegant interpretation involving nilpotent elements which break the bialgebra interaction while preserving all other rules in the ruleset. This establishes the first minimal ruleset for the pure-qubit ZX calculus. Furthermore, we establish the analogous results in a second ruleset, which combines the Hadamard decomposition and rotation rules into one, and which utilizes a different scalar cleanup rule. These results are the first of their kind for the pure-qubit ZX calculus, and resolve a nearly decade-long open problem.

\section{Background}
\label{sec:background}

In the full pure-qubit ZX calculus, a diagram \(D\) with \(k\) inputs and \(l\) outputs  \(D:k\to l\) is composed of the following generators:

\begingroup
\footnotesize
\setlength{\tabcolsep}{4pt}
\renewcommand{\arraystretch}{1.35}
\newcommand{\generatorBraidingHeight}{\vphantom{\vcenter{\hbox{\begin{tikzpicture}
	\begin{pgfonlayer}{nodelayer}
		\node [style=none] (0) at (0, 0) {};
		\node [style=none] (1) at (0.3, 0.3) {};
		\node [style=none] (2) at (0.3, -0.3) {};
		\node [style=none] (a) at (-0.3, 0.3) {};
		\node [style=none] (b) at (-0.3, -0.3) {};
		\node [style=none] (3) at (0.3, -0.5) {};
		\node [style=none] (4) at (0.3, 0.5) {};
	\end{pgfonlayer}
	\begin{pgfonlayer}{edgelayer}
		\draw[bend left=35] (1.center) to (0.center);
		\draw[bend right=35] (2.center) to (0.center);
		\draw[bend right=35] (a.center) to (0.center);
		\draw[bend left=35] (b.center) to (0.center);
	\end{pgfonlayer}
\end{tikzpicture}}}}}
\begin{center}
\begin{tabular}{|>{\centering\arraybackslash}p{0.23\textwidth}|>{\centering\arraybackslash}p{0.33\textwidth}|>{\centering\arraybackslash}p{0.34\textwidth}|}
\hline
Name & Diagram & Typed Version\\
\hline
Green Spider & \(\vcenter{\hbox{\begin{tikzpicture}\footnotesize
	\begin{pgfonlayer}{nodelayer}
		\node [style=none] (0) at (-0.5, 0.5) {};
		\node [style=none] (1) at (-1, 0.5) {};
		\node [style=none] (2) at (0, 0.45) {...};
		\node [style=none] (3) at (0, -0.45) {...};
		\node [style=none] (4) at (0.5, -0.5) {};
		\node [style=none] (5) at (0.75, 0.75) {};
		\node [style=none] (6) at (0.75, -0.75) {};
		\node [style=gn] (7) at (-0.25, -0) {$~\alpha~$};
		\node [style=none] (8) at (-1.25, 0.75) {};
		\node [style=none] (9) at (-1.25, -0.75) {};
		\node [style=none] (10) at (-1, -0.5) {};
		\node [style=none] (11) at (0.5, 0.5) {};
		\node [style=none] (12) at (-0.5, -0.5) {};
	\end{pgfonlayer}
	\begin{pgfonlayer}{edgelayer}
		\draw [style=braceedge] (6.center) to node[wire label, inner sep={5 pt}]{$m$} (9.center);
		\draw [style=none, bend left=15, looseness=1.00] (10.center) to (7);
		\draw [style=none, bend left=15, looseness=1.00] (12.center) to (7);
		\draw [style=braceedge] (8.center) to node[wire label, inner sep={5 pt}]{$n$} (5.center);
		\draw [style=none, bend left=15, looseness=1.00] (7) to (1.center);
		\draw [style=none, bend right=15, looseness=1.00] (4.center) to (7);
		\draw [style=none, bend right=15, looseness=1.00] (7) to (11.center);
		\draw [style=none, bend left=15, looseness=1.00] (7) to (0.center);
	\end{pgfonlayer}
\end{tikzpicture}}}\) & \(Z^\alpha_{n,m}:n\to m\)\\[0.5em]
\hline
Red Spider & \(\vcenter{\hbox{\begin{tikzpicture}\footnotesize
	\begin{pgfonlayer}{nodelayer}
		\node [style=none] (0) at (1.25, -0.5) {};
		\node [style=none] (1) at (2.25, 0.5) {};
		\node [style=none] (2) at (0.5, -0.75) {};
		\node [style=none] (3) at (0.75, 0.5) {};
		\node [style=none] (4) at (0.5, 0.75) {};
		\node [style=rn] (5) at (1.5, 0) {$~\alpha~$};
		\node [style=none] (6) at (2.5, -0.75) {};
		\node [style=none] (7) at (1.75, -0.45) {...};
		\node [style=none] (8) at (0.75, -0.5) {};
		\node [style=none] (9) at (2.25, -0.5) {};
		\node [style=none] (10) at (1.25, 0.5) {};
		\node [style=none] (11) at (1.75, 0.45) {...};
		\node [style=none] (12) at (2.5, 0.75) {};
	\end{pgfonlayer}
	\begin{pgfonlayer}{edgelayer}
		\draw [style=braceedge] (6.center) to node[wire label, inner sep={5 pt}]{$m$} (2.center);
		\draw [style=none, bend left=15, looseness=1.00] (8.center) to (5);
		\draw [style=none, bend left=15, looseness=1.00] (0.center) to (5);
		\draw [style=braceedge] (4.center) to node[wire label, inner sep={5 pt}]{$n$} (12.center);
		\draw [style=none, bend left=15, looseness=1.00] (5) to (3.center);
		\draw [style=none, bend right=15, looseness=1.00] (9.center) to (5);
		\draw [style=none, bend right=15, looseness=1.00] (5) to (1.center);
		\draw [style=none, bend left=15, looseness=1.00] (5) to (10.center);
	\end{pgfonlayer}
\end{tikzpicture}}}\) & \(X^\alpha_{n,m}:n\to m\)\\[0.5em]
\hline
Hadamard & \(\generatorBraidingHeight\vcenter{\hbox{\begin{tikzpicture}
	\begin{pgfonlayer}{nodelayer}
		\node [style={H box}] (0) at (0.5, 0) {};
		\node [style=none] (1) at (0.5, 0.3) {};
		\node [style=none] (2) at (0.5, -0.3) {};
				\node [style=none] (a) at (0.5, 0.5) {};
		\node [style=none] (b) at (0.5, -0.5) {};
	\end{pgfonlayer}
	\begin{pgfonlayer}{edgelayer}
		\draw (1.center) to (0);
		\draw (2.center) to (0);
	\end{pgfonlayer}
\end{tikzpicture}}}\) & \(H:1\to 1\)\\[0.5em]
\hline
Identity & \(\generatorBraidingHeight\vcenter{\hbox{\begin{tikzpicture}
	\begin{pgfonlayer}{nodelayer}
		\node [style=none] (1) at (0.5, 0.3) {};
		\node [style=none] (2) at (0.5, -0.3) {};
	\end{pgfonlayer}
	\begin{pgfonlayer}{edgelayer}
		\draw (1.center) to (2.center);
	\end{pgfonlayer}
\end{tikzpicture}}}\) & \(1_1:1\to 1\)\\[0.5em]
\hline
Braiding & \(\generatorBraidingHeight\vcenter{\hbox{\begin{tikzpicture}
	\begin{pgfonlayer}{nodelayer}
		\node [style=none] (0) at (0, 0) {};
		\node [style=none] (1) at (0.3, 0.3) {};
		\node [style=none] (2) at (0.3, -0.3) {};
		\node [style=none] (a) at (-0.3, 0.3) {};
		\node [style=none] (b) at (-0.3, -0.3) {};
		\node [style=none] (3) at (0.3, -0.5) {};
		\node [style=none] (4) at (0.3, 0.5) {};
	\end{pgfonlayer}
	\begin{pgfonlayer}{edgelayer}
		\draw[bend left=35] (1.center) to (0.center);
		\draw[bend right=35] (2.center) to (0.center);
		\draw[bend right=35] (a.center) to (0.center);
		\draw[bend left=35] (b.center) to (0.center);
	\end{pgfonlayer}
\end{tikzpicture}}}\) & \(\sigma:2\to 2\)\\[0.5em]
\hline
Cup & \(\generatorBraidingHeight\vcenter{\hbox{\begin{tikzpicture}
	\begin{pgfonlayer}{nodelayer}
		\node [style=none] (0) at (0, 0) {};
		\node [style=none] (1) at (0.3, 0.3) {};
		\node [style=none] (a) at (-0.3, 0.3) {};
	\end{pgfonlayer}
	\begin{pgfonlayer}{edgelayer}
		\draw[bend left=45] (1.center) to (0.center);
		\draw[bend right=45] (a.center) to (0.center);
	\end{pgfonlayer}
\end{tikzpicture}}}\) & \(\eta:0\to 2\)\\[0.5em]
\hline
Cap & \(\generatorBraidingHeight\vcenter{\hbox{\begin{tikzpicture}
	\begin{pgfonlayer}{nodelayer}
		\node [style=none] (0) at (0, 0) {};
		\node [style=none] (1) at (-0.3, -0.3) {};
		\node [style=none] (a) at (0.3, -0.3) {};
		\node [style=none] (3) at (0.2, -0.5) {};
		\node [style=none] (4) at (0.2, 0.5) {};
	\end{pgfonlayer}
	\begin{pgfonlayer}{edgelayer}
		\draw[bend left=45] (1.center) to (0.center);
		\draw[bend right=45] (a.center) to (0.center);
	\end{pgfonlayer}
\end{tikzpicture}}}\) & \(\epsilon:2\to 0\)\\[0.5em]
\hline
Unit (empty diagram) & \(\generatorBraidingHeight\vcenter{\hbox{
\begin{tikzpicture}
	\begin{pgfonlayer}{nodelayer}
		\node [style=none] (4) at (0.25, 0.25) {};
		\node [style=none] (5) at (-0.25, 0.25) {};
		\node [style=none] (6) at (0.25, -0.25) {};
		\node [style=none] (7) at (-0.25, -0.25) {};
	\end{pgfonlayer}
	\begin{pgfonlayer}{edgelayer}
		\draw [dashed, color=gray] (5.center) to (7.center);
		\draw [dashed, color=gray] (7.center) to (6.center);
		\draw [dashed, color=gray] (6.center) to (4.center);
		\draw [dashed, color=gray] (4.center) to (5.center);
	\end{pgfonlayer}
\end{tikzpicture}}}\) & \(e:0\to 0\)\\[0.5em]
\hline
\end{tabular}

\vspace{0.5em}

\hypertarget{tab:zx-generators}{\textbf{Figure 1.}} Here \(\alpha\in\mathbb{R}\), \(n,m\in\mathbb{N}\), and \(e\) is the empty diagram.
\end{center}
\endgroup

We can compose diagrams in two ways:

\begin{itemize}
\begin{samepage}
\item \textbf{Sequentially}: If \(D_1:a \to b\) and \(D_2:b \to c\) are two diagrams, we can compose \(D_2 \circ D_1 : a \to c\) by placing \(D_1\) above \(D_2\).  We can think of this new process as performing the process \(D_1\) first, and then applying the process \(D_2\) next.

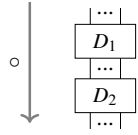
\widetikzitinput{"resources/figures/sequential-composition.tikz"}
\end{samepage}

\begin{samepage}
\item \textbf{Spatially}: If \(D_1:a\to b\) and \(D_2:c \to d\) are two diagrams, we can compose \(D_1 \ocross D_2: a+c \to b+d\) by placing \(D_1\) to the left of \(D_2\).  We can think of this new process as performing the processes \(D_1\) and \(D_2\) side-by-side.

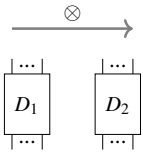
\widetikzitinput{"resources/figures/spatial-composition.tikz"}
\end{samepage}
\end{itemize}

These operations satisfy the coherence conditions for a strict monoidal category.

Importantly, we can define the \textbf{standard interpretation functor}:

For any ZX diagram \(D:n\to m\), \(\intf{D}:(\mathbb C^2)^{\otimes n}\to(\mathbb C^2)^{\otimes m}\) is defined inductively in the following way:

\[
\intf{D_1\otimes D_2}
=\intf{D_1}\otimes\intf{D_2},
\quad
\intf{D_2\circ D_1}
=\intf{D_2}\circ\intf{D_1}.
\]

\[
\begin{array}{rcl@{\qquad\qquad}rcl}
\intf{\vcenter{\hbox{\begin{tikzpicture}
	\begin{pgfonlayer}{nodelayer}
		\node [style=none] (1) at (0.5, 0.3) {};
		\node [style=none] (2) at (0.5, -0.3) {};
	\end{pgfonlayer}
	\begin{pgfonlayer}{edgelayer}
		\draw (1.center) to (2.center);
	\end{pgfonlayer}
\end{tikzpicture}}}}
&=&
\begin{pmatrix}
1 & 0\\
0 & 1
\end{pmatrix}
&
\intf{\vcenter{\hbox{\begin{tikzpicture}
	\begin{pgfonlayer}{nodelayer}
		\node [style=none] (0) at (0, 0) {};
		\node [style=none] (1) at (0.3, 0.3) {};
		\node [style=none] (2) at (0.3, -0.3) {};
		\node [style=none] (a) at (-0.3, 0.3) {};
		\node [style=none] (b) at (-0.3, -0.3) {};
		\node [style=none] (3) at (0.3, -0.5) {};
		\node [style=none] (4) at (0.3, 0.5) {};
	\end{pgfonlayer}
	\begin{pgfonlayer}{edgelayer}
		\draw[bend left=35] (1.center) to (0.center);
		\draw[bend right=35] (2.center) to (0.center);
		\draw[bend right=35] (a.center) to (0.center);
		\draw[bend left=35] (b.center) to (0.center);
	\end{pgfonlayer}
\end{tikzpicture}}}}
&=&
\begin{pmatrix}
1 & 0 & 0 & 0\\
0 & 0 & 1 & 0\\
0 & 1 & 0 & 0\\
0 & 0 & 0 & 1
\end{pmatrix}
\\[2em]
\intf{\vphantom{\vcenter{\hbox{\begin{tikzpicture}
	\begin{pgfonlayer}{nodelayer}
		\node [style=none] (0) at (0, 0) {};
		\node [style=none] (1) at (-0.3, -0.3) {};
		\node [style=none] (a) at (0.3, -0.3) {};
		\node [style=none] (3) at (0.2, -0.5) {};
		\node [style=none] (4) at (0.2, 0.5) {};
	\end{pgfonlayer}
	\begin{pgfonlayer}{edgelayer}
		\draw[bend left=45] (1.center) to (0.center);
		\draw[bend right=45] (a.center) to (0.center);
	\end{pgfonlayer}
\end{tikzpicture}}}}\vcenter{\hbox{\begin{tikzpicture}
	\begin{pgfonlayer}{nodelayer}
		\node [style=none] (0) at (0, 0) {};
		\node [style=none] (1) at (0.3, 0.3) {};
		\node [style=none] (a) at (-0.3, 0.3) {};
	\end{pgfonlayer}
	\begin{pgfonlayer}{edgelayer}
		\draw[bend left=45] (1.center) to (0.center);
		\draw[bend right=45] (a.center) to (0.center);
	\end{pgfonlayer}
\end{tikzpicture}}}}
&=&
\ket{00}+\ket{11}
&
\intf{\vcenter{\hbox{\begin{tikzpicture}
	\begin{pgfonlayer}{nodelayer}
		\node [style=none] (0) at (0, 0) {};
		\node [style=none] (1) at (-0.3, -0.3) {};
		\node [style=none] (a) at (0.3, -0.3) {};
		\node [style=none] (3) at (0.2, -0.5) {};
		\node [style=none] (4) at (0.2, 0.5) {};
	\end{pgfonlayer}
	\begin{pgfonlayer}{edgelayer}
		\draw[bend left=45] (1.center) to (0.center);
		\draw[bend right=45] (a.center) to (0.center);
	\end{pgfonlayer}
\end{tikzpicture}}}}
&=&
\bra{00}+\bra{11}
\end{array}
\]
\[
\intf{\vcenter{\hbox{\begin{tikzpicture}\footnotesize
	\begin{pgfonlayer}{nodelayer}
		\node [style=none] (0) at (-0.5, 0.5) {};
		\node [style=none] (1) at (-1, 0.5) {};
		\node [style=none] (2) at (0, 0.45) {...};
		\node [style=none] (3) at (0, -0.45) {...};
		\node [style=none] (4) at (0.5, -0.5) {};
		\node [style=none] (5) at (0.75, 0.75) {};
		\node [style=none] (6) at (0.75, -0.75) {};
		\node [style=gn] (7) at (-0.25, -0) {$~\alpha~$};
		\node [style=none] (8) at (-1.25, 0.75) {};
		\node [style=none] (9) at (-1.25, -0.75) {};
		\node [style=none] (10) at (-1, -0.5) {};
		\node [style=none] (11) at (0.5, 0.5) {};
		\node [style=none] (12) at (-0.5, -0.5) {};
	\end{pgfonlayer}
	\begin{pgfonlayer}{edgelayer}
		\draw [style=braceedge] (6.center) to node[wire label, inner sep={5 pt}]{$m$} (9.center);
		\draw [style=none, bend left=15, looseness=1.00] (10.center) to (7);
		\draw [style=none, bend left=15, looseness=1.00] (12.center) to (7);
		\draw [style=braceedge] (8.center) to node[wire label, inner sep={5 pt}]{$n$} (5.center);
		\draw [style=none, bend left=15, looseness=1.00] (7) to (1.center);
		\draw [style=none, bend right=15, looseness=1.00] (4.center) to (7);
		\draw [style=none, bend right=15, looseness=1.00] (7) to (11.center);
		\draw [style=none, bend left=15, looseness=1.00] (7) to (0.center);
	\end{pgfonlayer}
\end{tikzpicture}}}}
=
\ket{0}^{\otimes m}\bra{0}^{\otimes n}
+e^{i\alpha}\ket{1}^{\otimes m}\bra{1}^{\otimes n},
\qquad \alpha\in\mathbb R.
\]
\[
\intf{\vcenter{\hbox{\begin{tikzpicture}\footnotesize
	\begin{pgfonlayer}{nodelayer}
		\node [style=none] (0) at (1.25, -0.5) {};
		\node [style=none] (1) at (2.25, 0.5) {};
		\node [style=none] (2) at (0.5, -0.75) {};
		\node [style=none] (3) at (0.75, 0.5) {};
		\node [style=none] (4) at (0.5, 0.75) {};
		\node [style=rn] (5) at (1.5, 0) {$~\alpha~$};
		\node [style=none] (6) at (2.5, -0.75) {};
		\node [style=none] (7) at (1.75, -0.45) {...};
		\node [style=none] (8) at (0.75, -0.5) {};
		\node [style=none] (9) at (2.25, -0.5) {};
		\node [style=none] (10) at (1.25, 0.5) {};
		\node [style=none] (11) at (1.75, 0.45) {...};
		\node [style=none] (12) at (2.5, 0.75) {};
	\end{pgfonlayer}
	\begin{pgfonlayer}{edgelayer}
		\draw [style=braceedge] (6.center) to node[wire label, inner sep={5 pt}]{$m$} (2.center);
		\draw [style=none, bend left=15, looseness=1.00] (8.center) to (5);
		\draw [style=none, bend left=15, looseness=1.00] (0.center) to (5);
		\draw [style=braceedge] (4.center) to node[wire label, inner sep={5 pt}]{$n$} (12.center);
		\draw [style=none, bend left=15, looseness=1.00] (5) to (3.center);
		\draw [style=none, bend right=15, looseness=1.00] (9.center) to (5);
		\draw [style=none, bend right=15, looseness=1.00] (5) to (1.center);
		\draw [style=none, bend left=15, looseness=1.00] (5) to (10.center);
	\end{pgfonlayer}
\end{tikzpicture}}}}
=
\ket{+}^{\otimes m}\bra{+}^{\otimes n}
+e^{i\alpha}\ket{-}^{\otimes m}\bra{-}^{\otimes n},
\qquad \alpha\in\mathbb R.
\]
\[
\intf{\vcenter{\hbox{\begin{tikzpicture}
	\begin{pgfonlayer}{nodelayer}
		\node [style={H box}] (0) at (0.5, 0) {};
		\node [style=none] (1) at (0.5, 0.3) {};
		\node [style=none] (2) at (0.5, -0.3) {};
				\node [style=none] (a) at (0.5, 0.5) {};
		\node [style=none] (b) at (0.5, -0.5) {};
	\end{pgfonlayer}
	\begin{pgfonlayer}{edgelayer}
		\draw (1.center) to (0);
		\draw (2.center) to (0);
	\end{pgfonlayer}
\end{tikzpicture}}}}
=\frac{1}{\sqrt{2}}
\begin{pmatrix}
1 & 1\\
1 & -1
\end{pmatrix}
\]

\newpage
\subsection{The Optimized Calculus \texorpdfstring{\(\zxopt\)}{ZX\_opt}}

\begin{center}
\begingroup
\setlength{\tabcolsep}{4pt}
\renewcommand{\arraystretch}{1.25}
\resizebox{0.9\textwidth}{!}{%
\begin{tabular}{|c@{\qquad\quad~}c@{\qquad\quad~}c@{~}|}
\hline
&& \\
\sOneTableDiagram&\begin{tikzpicture}
	\begin{pgfonlayer}{nodelayer}
		\node [style=gn] (0) at (-0.75, 0) {};
		\node [style=none] (1) at (0, 0) {=};
		\node [style=none] (2) at (-0.75, 0.75) {};
		\node [style=none] (3) at (-0.75, -0.75) {};
		\node [style=none] (4) at (0.5, 0.75) {};
		\node [style=none] (5) at (0.5, -0.75) {};
		\node [style=none] (6) at (0, -0.25) {(I$_g$)};
	\end{pgfonlayer}
	\begin{pgfonlayer}{edgelayer}
		\draw (2.center) to (0);
		\draw (0) to (3.center);
		\draw (4.center) to (5.center);
	\end{pgfonlayer}
\end{tikzpicture}
&\begin{tikzpicture}
	\begin{pgfonlayer}{nodelayer}
		\node [style=rn] (0) at (-0.75, -0.3) {$\frac{\text{-}\pi}{4}$};
		\node [style=gn] (1) at (-0.75, 0.3) {$\frac{\pi}{4}$};
		\node [style=none] (2) at (0, 0) {=};
		\node [style=none] (3) at (0.75, 0.25) {};
		\node [style=none] (4) at (0.75, -0.25) {};
		\node [style=none] (5) at (1.25, 0.25) {};
		\node [style=none] (6) at (1.25, -0.25) {};
		\node [style=none] (7) at (0, -0.25) {(E)};
	\end{pgfonlayer}
	\begin{pgfonlayer}{edgelayer}
		\draw (0) to (1);
		\draw [style=dashed] (3.center) to (5.center);
		\draw [style=dashed] (5.center) to (6.center);
		\draw [style=dashed] (6.center) to (4.center);
		\draw [style=dashed] (4.center) to (3.center);
	\end{pgfonlayer}
\end{tikzpicture}\\
&& \\
\begin{tikzpicture}
	\begin{pgfonlayer}{nodelayer}
		\node [style=gn] (0) at (0.75, 0) {};
		\node [style=none] (1) at (2.25, -0.25) {};
		\node [style=none] (2) at (0.5, -0.5) {};
		\node [style=rn] (3) at (2.25, 0.25) {};
		\node [style=none] (4) at (1, -0.5) {};
		\node [style=rn] (5) at (0.75, 0.5) {};
		\node [style=rn] (6) at (2.75, 0.25) {};
		\node [style=none] (7) at (2.75, -0.25) {};
		\node [style=none] (8) at (1.5, 0) {$=$};
		\node [style=rn] (9) at (0, 0.25) {};
		\node [style=gn] (10) at (0, -0.25) {};
		\node [style=none] (11) at (1.5, -0.25) {(CP)};
	\end{pgfonlayer}
	\begin{pgfonlayer}{edgelayer}
		\draw [style=none] (5) to (0);
		\draw [style=none, bend right=23, looseness=1.00] (0) to (2.center);
		\draw [style=none, bend left=23, looseness=1.00] (0) to (4.center);
		\draw [style=none] (3) to (1.center);
		\draw [style=none] (6) to (7.center);
		\draw (9) to (10);
	\end{pgfonlayer}
\end{tikzpicture}&\begin{tikzpicture}
	\begin{pgfonlayer}{nodelayer}
		\node [style=none] (0) at (1.25, 0.75) {};
		\node [style=rn] (1) at (-1.5, -0.25) {};
		\node [style=none] (2) at (0.75, -0.75) {};
		\node [style=none] (3) at (-0.75, 1) {};
		\node [style=none] (4) at (0.75, 0.75) {};
		\node [style=none] (5) at (-1.5, -0.75) {};
		\node [style=none] (6) at (-1.5, 1) {};
		\node [style=gn] (7) at (-0.75, 0.5) {};
		\node [style=none] (8) at (0, -0) {$=$};
		\node [style=rn] (9) at (-0.75, -0.25) {};
		\node [style=gn] (10) at (1, -0.25) {};
		\node [style=gn] (11) at (-1.5, 0.5) {};
		\node [style=none] (12) at (-0.75, -0.75) {};
		\node [style=none] (13) at (1.25, -0.75) {};
		\node [style=rn] (14) at (1, 0.25) {};
		\node [style=rn] (15) at (-2, 0.25) {};
		\node [style=gn] (16) at (-2, -0.25) {};
		\node [style=none] (17) at (0, -0.25) {(B)};
	\end{pgfonlayer}
	\begin{pgfonlayer}{edgelayer}
		\draw [style=none] (12.center) to (9);
		\draw [style=none] (5.center) to (1);
		\draw [style=none] (7) to (3.center);
		\draw [style=none, bend right=23, looseness=1.00] (9) to (7);
		\draw [style=none] (11) to (6.center);
		\draw [style=none, bend left=23, looseness=1.00] (1) to (11);
		\draw [style=none, bend right=23, looseness=1.00] (13.center) to (10);
		\draw [style=none] (10) to (14);
		\draw [style=none, bend left=23, looseness=1.00] (14) to (4.center);
		\draw [style=none, bend right=23, looseness=1.00] (14) to (0.center);
		\draw [bend right=23, looseness=1.00] (10) to (2.center);
		\draw (11) to (9);
		\draw (7) to (1);
		\draw (15) to (16);
	\end{pgfonlayer}
\end{tikzpicture}&\multirow{4}{*}{\begin{tikzpicture}
	\begin{pgfonlayer}{nodelayer}
		\node [style=none] (0) at (0, 0) {=};
		\node [style=gn] (1) at (0.75, 0) {$\beta_2$};
		\node [style=rn] (2) at (0.75, 0.55) {$\beta_1$};
		\node [style=rn] (3) at (0.75, -0.55) {$\beta_3$};
		\node [style=rn] (4) at (-0.75, 0) {$\alpha_2$};
		\node [style=gn] (5) at (-0.75, 0.5) {$\alpha_1$};
		\node [style=gn] (6) at (-0.75, -0.5) {$\alpha_3$};
		\node [style=none] (7) at (0.75, 1) {};
		\node [style=none] (8) at (0.75, -1) {};
		\node [style=none] (9) at (-0.75, 1) {};
		\node [style=none] (10) at (-0.75, -1) {};
		\node [style=rn] (11) at (1.5, -0.25) {$\pi$};
		\node [style=gn] (12) at (1.5, -0.75) {$\gamma$};
		\node [style=rn] (13) at (1.5, 0.75) {};
		\node [style=gn] (14) at (1.5, 0.25) {};
		\node [style=none] (15) at (0, -0.25) {(EU)};
	\end{pgfonlayer}
	\begin{pgfonlayer}{edgelayer}
		\draw [style=none] (7.center) to (8.center);
		\draw [style=none] (9.center) to (10.center);
		\draw [style=none] (12) to (11);
		\draw [style=none, bend left=45, looseness=1.25] (14) to (13);
		\draw [style=none, bend right=45, looseness=1.25] (14) to (13);
		\draw [style=none] (13) to (14);
	\end{pgfonlayer}
\end{tikzpicture}}\\
&& \\
\begin{tikzpicture}
	\begin{pgfonlayer}{nodelayer}
		\node [style=gn] (0) at (0.5, -0.55) {$\frac{\pi}{2}$};
		\node [style=rn] (1) at (0.5, -0) {};
		\node [style=gn] (2) at (0.5, 0.55) {$\frac{\pi}{2}$};
		\node [style=none] (3) at (0.5, 0.9999999) {};
		\node [style=gn] (4) at (1.25, 0.5) {$\frac{\text{-}\pi}{2}$};
		\node [style=none] (5) at (0.5, -0.9999999) {};
		\node [style=none] (6) at (-0.9999999, 0.9999999) {};
		\node [style=none] (7) at (-0.9999999, -0.9999999) {};
		\node [style=none] (8) at (-0.2500001, -0) {$=$};
		\node [style={{H box}}] (9) at (-0.9999999, -0) {};
		\node [style=none] (10) at (-0.25, -0.25) {(HD)};
	\end{pgfonlayer}
	\begin{pgfonlayer}{edgelayer}
		\draw (3.center) to (2);
		\draw (2) to (1);
		\draw (1) to (0);
		\draw (0) to (5.center);
		\draw (1) to (4);
		\draw (6.center) to (7.center);
	\end{pgfonlayer}
\end{tikzpicture}&\begin{tikzpicture}
	\begin{pgfonlayer}{nodelayer}
		\node [style=none] (0) at (0.7500001, 1) {};
		\node [style=none] (1) at (-0.7500001, -1) {};
		\node [style=H box] (2) at (-0.7500001, 0.5000001) {};
		\node [style=rn] (3) at (-1.25, -0) {\footnotesize$~\alpha~$};
		\node [style=H box] (4) at (-1.75, -0.5000001) {};
		\node [style=none] (5) at (-0.7500001, 1) {};
		\node [style=H box] (6) at (-0.7500001, -0.5000001) {};
		\node [style=none] (7) at (-1.25, 0.7500001) {$\cdots$};
		\node [style=none] (8) at (0, -0) {$=$};
		\node [style=gn] (9) at (1.25, -0) {\footnotesize$~\alpha~$};
		\node [style=none] (10) at (1.75, -1) {};
		\node [style=none] (11) at (1.75, 1) {};
		\node [style=none] (12) at (-1.75, 1) {};
		\node [style=H box] (13) at (-1.75, 0.5000001) {};
		\node [style=none] (14) at (-1.75, -1) {};
		\node [style=none] (15) at (0.7500001, -1) {};
		\node [style=none] (16) at (-1.25, -0.7500001) {$\cdots$};
		\node [style=none] (17) at (1.25, 0.7500001) {$\cdots$};
		\node [style=none] (18) at (1.25, -0.7500001) {$\cdots$};
		\node [style=none] (19) at (0, -0.25) {(H)};
	\end{pgfonlayer}
	\begin{pgfonlayer}{edgelayer}
		\draw [bend right, looseness=1.00] (3) to (4);
		\draw [bend left, looseness=1.00] (3) to (6);
		\draw (6) to (1.center);
		\draw (4) to (14.center);
		\draw [bend left, looseness=1.00] (3) to (13);
		\draw [bend right, looseness=1.00] (3) to (2);
		\draw (2) to (5.center);
		\draw (13) to (12.center);
		\draw [bend left=23, looseness=1.00] (9) to (0.center);
		\draw [bend right=23, looseness=1.00] (9) to (11.center);
		\draw [bend right=23, looseness=1.00] (9) to (15.center);
		\draw [bend left=23, looseness=1.00] (9) to (10.center);
	\end{pgfonlayer}
\end{tikzpicture}&\\
&& \\
\hline
\end{tabular}}

\vspace{0.5em}

\begin{minipage}{\textwidth}
\small
\textbf{Figure 2.} Set of rules \zxopt for the ZX calculus with scalars. The above labels and rewrites are lightly adapted from Vilmart~\cite{8785765}. The right-hand side of (E) is an empty diagram. (...) denote zero or more wires, while (\rotatebox{45}{\raisebox{-0.4em}{$\cdots$}}) denote one or more wires. In rule (EU), $\beta_1,\beta_2,\beta_3$ and $\gamma$ can be determined as follows: $x^+:=\frac{\alpha_1+\alpha_3}{2}$, $x^-:=x^+-\alpha_3$, $z := \cos{\frac{\alpha_2}{2}}\cos{x^+}+i\sin{\frac{\alpha_2}{2}}\cos{x^-}$ and $z' := \cos{\frac{\alpha_2}{2}}\sin{x^+}-i\sin{\frac{\alpha_2}{2}}\sin{x^-}$, then $\beta_1 = \arg z + \arg z', \beta_2 = 2\arg\left(i+\left|\frac{z}{z'}\right|\right), \beta_3 = \arg z - \arg z', \gamma = x^+-\arg(z)+\frac{\alpha_2-\beta_2}{2}$ where by convention $\arg(0):=0$ and $z'=0\implies \beta_2=0$.
\end{minipage}
\endgroup
\end{center}

In addition to these collections of rules, we also consider the connectivity meta-rule, defined below.  In practice, this intuitive notion lets one rearrange a ZX diagram in any way, as long as each component retains the original connectivity.  Indeed, this is a feature of PROPs (product and permutation categories) as well as other ``nice'' monoidal categories with their associated string diagrams.  The interested reader is referred to Heunen and Vicary~\cite{10.1093/oso/9780198739623.001.0001}.

\begin{definition}[Meta-rule: only connectivity matters]
\label{def:connectivity-meta-rule}
Two diagrams represent the same matrix whenever one can be transformed into the other by moving components around without changing their connections.
\end{definition}

Note that we are only assuming the above ruleset and the connectivity meta-rule.  No color-swapped or flipped versions of the rules are assumed a priori.

Now that we have all the ingredients, there are some key features one would want to capture in a graphical theory \(\mathrm{ZX}\) (with respect to pure qubit quantum mechanics):

\begin{itemize}
\item \textbf{Universality.}
\[
\forall A\in \mathbb C^{2^m\times 2^n},\qquad \exists D:n\to m,\qquad \intf{D}=A.
\]
\item \textbf{Soundness.}
\[
\forall D_1,D_2,\qquad \zxopt\vdash D_1=D_2\implies \intf{D_1}=\intf{D_2}.
\]
\item \textbf{Completeness.}
\[
\forall D_1,D_2,\qquad \intf{D_1}=\intf{D_2}\implies \zxopt\vdash D_1=D_2.
\]
\item \textbf{Minimality.}
\[
\forall R\in \zxopt,\qquad \zxopt\setminus\{R\}\nvdash R.
\]
\end{itemize}

Note that these can be viewed as properties of the interpretation functor.  Moreover, universality is a feature of the generators, while Soundness, Completeness and Minimality deal intimately with the rewrite system.

We refer to the ruleset in Figure~2 as \(\zxopt\), and to Vilmart's original version as \(\zxv\). Thus,
\[
\zxopt=\zxv\setminus\{(I_r)\}.
\]
Vilmart~\cite{8785765} proved that \(\zxv\) satisfies completeness and the necessity of most of the ruleset, with the exception of \((B)\), \((I_g)\), and \hyperref[eq:ir-rule-derivable]{\((I_r)\)}.  Notably, this broadly matches the situation of BPW~\cite{Backens_2020}, where completeness of their ruleset was established for the stabilizer ZX calculus, with the status of the bialgebra and \((I)\)-equivalent rules left open.  The necessity of these rules was shown in an earlier paper by the author~\cite{stoltz2026simplified}.  This conceptual symmetry motivates the use of similar techniques in this paper.

Now, we collect some important facts and theorems established in Vilmart~\cite{8785765}.

\begin{theorem}[Vilmart~\cite{8785765}]
\(\zxv\) is sound and complete for pure qubit quantum mechanics. For any two ZX diagrams \(D_1\) and \(D_2\),
\[
\intf{D_1}=\intf{D_2}
\quad\Longleftrightarrow\quad
\zxv\vdash D_1=D_2.
\]
\end{theorem}

\begin{theorem}[Vilmart~\cite{8785765}]
The rules of \(\zxv\), except possibly \((B)\) and the identity-rule pair,
are necessary:
\[
\forall R\in \zxv\setminus\{(B),(I_g),(I_r)\},\qquad
\zxv\setminus\{R\}\nvdash R.
\]
Moreover, at least one of the identity rules is necessary:
\[
\exists R\in\{(I_g),(I_r)\},\qquad
\zxv\setminus\{R\}\nvdash R.
\]
\end{theorem}

We show that the \((I_r)\) rule is derivable from \(\zxopt\), hence our ruleset is complete:

\begin{theorem}
\label{lem:ir-rule-derivable}
The \hyperref[eq:ir-rule-derivable]{\((I_r)\)} rule is derivable:
\phantomsection\label{eq:ir-rule-derivable}
\begin{equation*}
\zxopt\vdash
\vcenter{\hbox{\hkstikzitinput{HKS/figures/ir-derivation/ir-atomic-ir-statement.tikz}}}
\end{equation*}
\end{theorem}

\begin{irproof}
\hyperref[app:ir-rule-derivation]{See Appendix A for details of this derivation.}
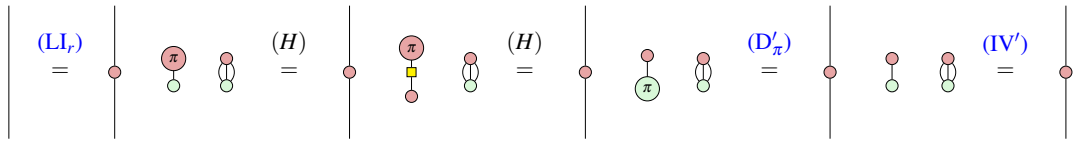
\begin{center}
\makebox[\linewidth][c]{%
	\hkstikzitinput{HKS/figures/ir-derivation/renaud-short/ir-renaud-ir-proof-a.tikz}\hspace{1em}%
	\hkstikzitinput{HKS/figures/ir-derivation/renaud-short/ir-renaud-ir-proof-b.tikz}\hspace{1em}%
	\hkstikzitinput{HKS/figures/ir-derivation/renaud-short/ir-renaud-ir-proof-c.tikz}%
}
\end{center}
\end{irproof}

Before the results of this paper, the remaining unresolved minimality questions were the necessity of \((B)\) and the individual status of the two identity rules \((I_g)\) and \((I_r)\). We have shown \((I_r)\) derivable, and will further prove the necessity of \((B)\) and \((I_g)\). Thus,

\begin{theorem}
\label{thm:zxopt-complete-minimal}
\(\zxopt\) is a complete and minimal ruleset for the pure-qubit ZX calculus.
\end{theorem}

\begin{irproof}
This follows from the results above and
\hyperref[sec:necessity-bialgebra-b-rule]{Sections~\ref*{sec:necessity-bialgebra-b-rule}}
and~\hyperref[sec:necessity-green-identity-ig-rule]{\ref*{sec:necessity-green-identity-ig-rule}}.
\end{irproof}

Additionally, we adapt a second ruleset from Vilmart~\cite{8785765}, \(\zxvprime\), to create a second complete, minimal ruleset \(\zxoptprime=\zxvprime\setminus\{(I_r)\}\). The proof of this statement is found in \hyperref[sec:another-optimized-ruleset]{Section~\ref*{sec:another-optimized-ruleset}}.

\section{On the Necessity of the Bialgebra (B) rule}
\label{sec:necessity-bialgebra-b-rule}

\subsection{Defining the target PROP \texorpdfstring{\(\mathcal E\)}{E}}
\label{sec:defining-target-category-e}

Let
\[
D:=\mathbb C[\epsilon]/(\epsilon^2)
\]
and let
\[
V:=D^{\{0,1\}^3}.
\]
We use basis vectors \(\lvert x\rangle\) for strings \(x\in\{0,1\}^3\), with Hamming weight \(\lvert x\rvert\).

Now, define a PROP \(\mathcal E\) where we interpret \(V^{\otimes 0}:=D\), and
\[
\mathcal E(n,m):=\operatorname{Hom}_D(V^{\otimes n},V^{\otimes m}).
\]
Sequential composition is given by the usual composition of \(D\)-linear maps.
Tensor product is given by the usual Kronecker tensor product with respect to our basis \(\lvert x\rangle\), \(x\in\{0,1\}^3\).

Finally, we take our structural generators to be the usual ones:
\[
I=\operatorname{id}_V,\qquad
\sigma(\lvert x\rangle\otimes\lvert y\rangle)
=
\lvert y\rangle\otimes\lvert x\rangle,
\]
\[
\eta(1)=\sum_{x\in\{0,1\}^3}\lvert x\rangle\otimes\lvert x\rangle,
\qquad
\varepsilon(\lvert x\rangle\otimes\lvert y\rangle)=\delta_{x,y},
\qquad
e=1\in D.
\]
For computational ease and clarity of matrix elements, we use the basis vectors
\[
\lvert x\rangle,\qquad x\in\{0,1\}^3,
\]
ordered lexicographically when writing a matrix.

Now, let us define the interpreted Z-spiders, Hadamard, and subsequently X-spiders built from them.

\subsection{Defining the Hadamard in \texorpdfstring{\(\mathcal E\)}{E}}
\label{sec:defining-hadamard-e}

By previous work on minimality in the stabilizer fragment, we want \(H'\) to satisfy:
\begin{equation}
\label{eq:hadamard-target}
(H')^T=H' \quad \& \quad (H')^2=I.
\end{equation}
In Stoltz~\cite{stoltz2026simplified}, this matrix was built using a ``nilpotent leak'' in the underlying ring, in such a way that it retained properties~\eqref{eq:hadamard-target}, which was enough to satisfy all non-target rules but break the bialgebra.
This time, we perturb the usual three-fold Hadamard by conjugation.
\[
h=\frac{1}{\sqrt 2}
\begin{pmatrix}
1&1\\
1&-1
\end{pmatrix},
\qquad
H_0:=h^{\otimes 3}.
\]
Note that our unperturbed Hadamard \(H_0\) is the \(8\)-dimensional Sylvester--Walsh Hadamard matrix.

Now, let us define a matrix \(N\), with the goal of skew-rotating the weight-one subspace defined by the basis vectors \(\lvert 100\rangle,\lvert 010\rangle,\lvert 001\rangle\). Let \(N\) act only on these basis vectors and nowhere else:
\[
N\lvert 100\rangle=\lvert 010\rangle-\lvert 001\rangle,
\]
\[
N\lvert 010\rangle=\lvert 001\rangle-\lvert 100\rangle,
\]
\[
N\lvert 001\rangle=\lvert 100\rangle-\lvert 010\rangle,
\]
and
\[
N\lvert x\rangle=0
\]
for all other \(x\in\{0,1\}^3\).

Then define
\[
O:=I+\epsilon N.
\]
This should be thought of as preserving all vectors in \(\{0,1\}^3\), except adding a small \(\epsilon\)-skew rotation from \(N\) on the basis vectors \(\lvert 100\rangle,\lvert 010\rangle,\lvert 001\rangle\).
Then define
\[
H':=O^{-1}H_0O.
\]
Note that
\[
O^{-1}=I-\epsilon N
\]
since \(\epsilon^2=0\). Hence
\[
H'=H_0+\epsilon(H_0N-NH_0).
\]

Now, since \(N^T=-N\), we have
\[
O^T=O^{-1}.
\]
Therefore
\[
(H')^T
=
(O^{-1}H_0O)^T
=
O^TH_0^T(O^{-1})^T
=
O^{-1}H_0O
=
H'.
\]
Also,
\[
(H')^2
=
O^{-1}H_0OO^{-1}H_0O
=
O^{-1}H_0^2O
=
I,
\]
as desired.
Now that we have constructed such a Hadamard, we will see that ``only connectivity matters'' will hold (Lemma~\ref{lem:connectivity-soundness}).

\subsection{Defining the Green Spiders in \texorpdfstring{\(\mathcal E\)}{E}}
\label{sec:defining-green-spiders-e}

Now, let us turn our focus to the interpretation of the Z-spiders.

In Stoltz~\cite{stoltz2026simplified}, Z-spiders were defined as the usual basis-copying tensor spiders, decorated by phases from a tuple \(t\).

We want a tuple \(t_x\), \(x\in\{0,1\}^3\), such that spider fusion becomes pointwise multiplication of decorations
\begin{equation}
\label{eq:phase-character}
 t_x(\alpha)t_x(\beta)=t_x(\alpha+\beta).
\end{equation}
These can be taken to be characters of the additive phase group \((\mathbb R,+)\), which settles this property precisely.

Define
\[
t_x(\alpha):=e^{i\lvert x\rvert\alpha},
\qquad x\in\{0,1\}^3.
\]
This parametrized phase family satisfies~\eqref{eq:phase-character}.
Also note that the usual qubit phase gate
\[
P(\alpha)=
\begin{pmatrix}
1&0\\
0&e^{i\alpha}
\end{pmatrix}
\]
has the property that \(P(\alpha)^{\otimes 3}\) is diagonal on the basis \(\lvert x\rangle\), with diagonal entry \(e^{i\lvert x\rvert\alpha}\).

\noindent\textbf{Definition.} Now, with this motivation, we define the interpreted Z-spider:
\[
(Z')^\alpha_{n,m}
:=
\sum_{x\in\{0,1\}^3}
e^{i\lvert x\rvert\alpha}
\lvert x\rangle^{\otimes m}
\langle x\rvert^{\otimes n}.
\]

\subsection{Defining the Red Spiders in \texorpdfstring{\(\mathcal E\)}{E}}
\label{sec:defining-red-spiders-e}

\noindent\textbf{Definition.} Finally, we define the interpreted X-spider in the usual way:
\[
 (X')^\alpha_{n,m}:=(H')^{\otimes m}\circ (Z')^\alpha_{n,m}\circ (H')^{\otimes n}.
\]

\subsection{Defining the Interpretation Functor}
\label{sec:defining-interpretation-functor}

Define the \(D\)-linear interpretation
\[
\sem{-}^{(B)}:V^{\otimes n}\longrightarrow V^{\otimes m}
\]
inductively as follows:
\[
\sem{D_1\otimes D_2}^{(B)}=\sem{D_1}^{(B)}\otimes \sem{D_2}^{(B)},
\qquad
\sem{D_2\circ D_1}^{(B)}=\sem{D_2}^{(B)}\circ \sem{D_1}^{(B)}.
\]
\[
\sem{\zxfig{gn-alpha}}^{(B)}
= (Z')^\alpha_{n,m}:n\to m.
\]
\[
\sem{\zxfig{Hadamard}}^{(B)}=H':1\to1.
\]
\[
\sem{\zxfig{rn-alpha}}^{(B)}
= (X')^\alpha_{n,m}:n\to m.
\]

All structural generators are interpreted as the standard versions outlined above.

\begin{lemma}
\label{lem:connectivity-soundness}
\(\sem{-}^{(B)}\) satisfies the ``only connectivity matters'' meta-rule.
\end{lemma}

\begin{proof}
Since our structural generators are standard with respect to the basis \(\{\lvert x\rangle:x\in\{0,1\}^3\}\), we have the yanking/snake equations:
\[
\forall x\in\{0,1\}^3,\qquad
(\varepsilon\otimes I)(I\otimes\eta)(\lvert x\rangle)
=
\lvert x\rangle
=
(I\otimes\varepsilon)(\eta\otimes I)(\lvert x\rangle).
\]
The braiding \(\sigma\) is symmetric too.

Now, the Hadamard was chosen such that
\[
(H')^T=H' \qquad\text{and}\qquad (H')^2=I,
\]
so Hadamards can be bent through the compact structure.

As for the green spiders, \((Z')^\alpha_{n,m}\) is flexsymmetric: it is invariant under permutation of its input and output wires.
Also, one can check that
\[
\bigl((Z')^\alpha_{n,m}\bigr)^T=(Z')^\alpha_{m,n}.
\]
Finally, the red spiders inherit these properties by definition.

So, any diagram constructed from our generators depends only on their connectivity.
\end{proof}

\subsection{Necessity of the Bialgebra}
\label{sec:necessity-of-the-bialgebra}

\begin{theorem}
\label{thm:b-rule-necessity}
The \((B)\) rule is necessary:
\[
\zxopt\setminus\{(B)\}\nvdash (B).
\]
\end{theorem}

\noindent\textbf{Proof.} All rules but \((B)\) are sound with respect to the interpretation \(\sem{-}^{(B)}\).
See \hyperref[app:b-rule-soundness-checks]{Appendix B} for soundness checks.

It is enough to find one basis vector and one coefficient on which the two sides disagree.
In the case of \((B)\), evaluating both sides on
\[
\lvert 010\rangle\otimes\lvert 010\rangle
\]
and comparing the
\[
\lvert 000\rangle\otimes\lvert 001\rangle
\]
coefficient gives
\begingroup
\small
\[
\begin{aligned}
(\bra{000}\otimes\bra{001})\cdot
\Bigl(\sem{\zxbig{B-LHS-only}}^{(B)}(\ket{010}\otimes\ket{010})\Bigr)
&=-\frac{\sqrt 2}{4}\epsilon,\\
(\bra{000}\otimes\bra{001})\cdot
\Bigl(\sem{\zxbig{B-RHS-only}}^{(B)}(\ket{010}\otimes\ket{010})\Bigr)
&=0.
\end{aligned}
\]
\endgroup
Clearly these matrices are not the same. Hence, the bialgebra rule \((B)\) is necessary!
\[
\text{Thus, }\zxopt\setminus\{(B)\}\nvdash(B).
\]
\tombstone

\section{On the Necessity of the Green Identity \texorpdfstring{\((I_g)\)}{(Ig)} rule}
\label{sec:necessity-green-identity-ig-rule}

We will introduce a relational interpretation
\(\sem{-}^{\mathrm{rel}}\) devised by BPW~\cite{Backens_2020}, and extend it
to a new one, \(\sem{-}^{(I_g)}\) to break \((I_g)\).
This technique avoids scalar normalization altogether.

\bigskip

Given a relation \(R:A\to B\), we can
assign it a corresponding logical matrix \(M_R\),
where rows are indexed by elements of \(A\),
and columns are indexed by elements of \(B\), by
convention,
\begin{itemize}
    \item If elements \(a_i\sim b_j\), then \((M_R)_{ij}=1\).
    \item If elements \(a_i\nsim b_j\), then \((M_R)_{ij}=0\).
\end{itemize}

Then, for two logical matrices \(M,N\) over \(\bool\), we have two standard ways of composing them:
\begin{itemize}
\item \textbf{sequential composition.}

If \(M:\bool^r\to\bool^s\) and \(N:\bool^s\to\bool^t\), then
\[
N\circ M:\bool^r\to\bool^t,
\qquad
(N\circ M)_{a,c}:=\bigvee_{b\in \bool^s}(M_{a,b}\wedge N_{b,c})
\]
\[
(\text{where }a\in\bool^r,\ c\in\bool^t,\text{ and }\vee,\wedge\text{ denote logical OR and AND, respectively})
\]

\item \textbf{spatial composition.}

If \(M:\bool^r\to\bool^s\) and \(N:\bool^{r'}\to\bool^{s'}\), then
\[
M\otimes N:\bool^{r+r'}\to\bool^{s+s'},
\qquad
(M\otimes N)_{(a,c),(b,d)}:=M_{a,b}\wedge N_{c,d}.
\]
\[
(\text{this is the usual Kronecker tensor product})
\]
\end{itemize}

Now, for a diagram \(D:n\to m\), we define
\[
\sem{D}^{\mathrm{rel}}:\bool^n\to\bool^m
\]
inductively as follows:
\[
\sem{D_1\otimes D_2}^{\mathrm{rel}}=\sem{D_1}^{\mathrm{rel}}\otimes\sem{D_2}^{\mathrm{rel}}
\qquad
\sem{D_2\circ D_1}^{\mathrm{rel}}=\sem{D_2}^{\mathrm{rel}}\circ\sem{D_1}^{\mathrm{rel}}.
\]
\[
\sem{\zxfig{empty-diagram}}^{\mathrm{rel}}=1,
\qquad
\sem{\zxfig{single-line}}^{\mathrm{rel}}=
\begin{pmatrix}1&0\\0&1\end{pmatrix},
\qquad
\sem{\zxfig{crossing}}^{\mathrm{rel}}=
\begin{pmatrix}
1&0&0&0\\
0&0&1&0\\
0&1&0&0\\
0&0&0&1
\end{pmatrix}.
\]
\[
\sem{\zxfig{cup}}^{\mathrm{rel}}=(1\ 0\ 0\ 1),
\qquad
\sem{\zxfig{caps}}^{\mathrm{rel}}=
\begin{pmatrix}1\\0\\0\\1\end{pmatrix}.
\]
These are our structural generators.

\bigskip

Now, we will extend this setup to \(\sem{-}^{(I_g)}\), our
new interpretation.

\bigskip

First, let us examine the properties we want to
preserve and break.

\bigskip

Given spiders \(Z^\alpha_{n,m}:n\to m\) and \(X^\alpha_{n,m}:n\to m\),
we write \(d(n,m):=n+m\) for the total degree.
\(d(n,m)=n+m=2\) is a special case, since this
matches the arities of the compact structure: \((1,1),(0,2),(2,0)\).

\bigskip

We want to isolate and separate the case
when \(\alpha=0\) and \(n+m=2\).

\bigskip

In this case, we want three things:
\begin{itemize}
    \item For the connectivity meta-rule to hold.
    \item For the \((I_r)\) rule to hold: zero phase degree 2 red spiders must behave like a wire, cap, and cup.
    \item For the \((I_g)\) rule to break.
\end{itemize}

Now, whenever \(n+m=2\), we write the \textbf{two-point relation}:
\[
C_{n,m}:=\{(0^n,0^m),(1^n,1^m)\}\subseteq \bool^n\times\bool^m.
\]
This gives three cases:
\[
C_{1,1}=
\begin{pmatrix}1&0\\0&1\end{pmatrix}
=\sem{\zxfig{single-line}}^{\mathrm{rel}}
\]
\[
C_{0,2}=(1\ 0\ 0\ 1)
=\sem{\zxfig{cup}}^{\mathrm{rel}}
\]
\[
C_{2,0}=
\begin{pmatrix}1\\0\\0\\1\end{pmatrix}
=\sem{\zxfig{caps}}^{\mathrm{rel}}
\]

For arbitrary \(n,m\), we write the \textbf{one-point relation}:
\[
P_{n,m}:=\{(0^n,0^m)\}\subseteq \bool^n\times\bool^m.
\]
In the \(n+m=2\) case, we get:
\[
P_{1,1}=
\begin{pmatrix}1&0\\0&0\end{pmatrix}
\ne \sem{\zxfig{single-line}}^{\mathrm{rel}},
\]
\[
P_{0,2}=(1\ 0\ 0\ 0)
\ne \sem{\zxfig{cup}}^{\mathrm{rel}},
\]
\[
P_{2,0}=
\begin{pmatrix}1\\0\\0\\0\end{pmatrix}
\ne \sem{\zxfig{caps}}^{\mathrm{rel}}.
\]

For both \(C_{n,m}\) and \(P_{n,m}\) we have:
\[
P_{n,m}^T=P_{m,n}\quad \& \quad C_{n,m}^T=C_{m,n}
\]
so these are good candidates.

Hence, for a diagram \(D:n\to m\), we will define the interpretation
\[
\sem{D}^{(I_g)}:\bool^n\to\bool^m
\]
inductively:
\[
\sem{D_1\otimes D_2}^{(I_g)}=\sem{D_1}^{(I_g)}\otimes\sem{D_2}^{(I_g)},
\qquad
\sem{D_2\circ D_1}^{(I_g)}=\sem{D_2}^{(I_g)}\circ\sem{D_1}^{(I_g)}.
\]
\[
\sem{\zxfig{empty-diagram}}^{(I_g)}=\sem{\zxfig{empty-diagram}}^{\mathrm{rel}},
\qquad
\sem{\zxfig{single-line}}^{(I_g)}=\sem{\zxfig{single-line}}^{\mathrm{rel}},
\qquad
\sem{\zxfig{crossing}}^{(I_g)}=\sem{\zxfig{crossing}}^{\mathrm{rel}},
\]
\[
\sem{\zxfig{cup}}^{(I_g)}=\sem{\zxfig{cup}}^{\mathrm{rel}},
\qquad
\sem{\zxfig{caps}}^{(I_g)}=\sem{\zxfig{caps}}^{\mathrm{rel}}.
\]

We will ``break'' the green spiders and Hadamard together:
\[
\sem{\zxfig{gn-alpha}}^{(I_g)}=P_{n,m},
\qquad
\sem{\zxfig{Hadamard}}^{(I_g)}=P_{1,1}.
\]
And modify the red spider with the exception:
\[
\sem{\zxfig{rn-alpha}}^{(I_g)}=
\begin{cases}
C_{n,m} & \text{, }\alpha=0\text{ and }n+m=2,\\
P_{n,m} & \text{else.}
\end{cases}
\]

Note that when we are not in the \(\alpha=0\) and \(n+m=2\) case:
\[
\sem{\zxfig{rn-alpha}}^{(I_g)}=
\left(\sem{\zxfig{Hadamard}}^{(I_g)}\right)^{\otimes m}
\circ \sem{\zxfig{gn-alpha}}^{(I_g)}
\circ \left(\sem{\zxfig{Hadamard}}^{(I_g)}\right)^{\otimes n}.
\]

\bigskip

\refstepcounter{theorem}
\noindent\textbf{Lemma~\thetheorem.} \(\sem{-}^{(I_g)}\) satisfies the ``only connectivity
matters'' meta-rule.

\medskip

\noindent\textbf{Proof.} Since our structural generators are
standard with respect to the usual ones,
the yanking / snake equations hold and
we have the usual symmetry of the braiding.

Now, since \(P_{n,m}^T=P_{m,n}\) and \(C_{n,m}^T=C_{m,n}\), and since both \(P_{n,m}\) and \(C_{n,m}\) are invariant
under permutations of their inputs / outputs, we get:
\[
\left(\sem{\zxfig{gn-alpha}}^{(I_g)}\right)^T
=P_{n,m}^T=P_{m,n}=\sem{Z^\alpha_{m,n}}^{(I_g)}
\]
\[
\left(\sem{\zxfig{Hadamard}}^{(I_g)}\right)^T
=P_{1,1}^T=P_{1,1}=\sem{\zxfig{Hadamard}}^{(I_g)}
\]
\[
\left(\sem{\zxfig{rn-alpha}}^{(I_g)}\right)^T
=
\begin{cases}
C_{n,m}^T & \text{, }\alpha=0\text{ and }n+m=2,\\
P_{n,m}^T & \text{else}
\end{cases}
=
\begin{cases}
C_{m,n} & \text{, }\alpha=0\text{ and }n+m=2,\\
P_{m,n} & \text{else}
\end{cases}
=\sem{X^\alpha_{m,n}}^{(I_g)}.
\]
Therefore, any diagram built from our
generators in the relational PROP depends
only on their connectivity.
\tombstone

\bigskip

\begin{theorem}
\label{thm:ig-rule-necessity}
The \((I_g)\) rule is necessary:
\[
\zxopt\setminus\{(I_g)\}\nvdash (I_g)
\]
\end{theorem}

\noindent\textbf{Proof.} All rules but \((I_g)\) are sound with respect to
the interpretation \(\sem{-}^{(I_g)}\).
See \hyperref[app:ig-rule-soundness-checks]{Appendix B} for details.

Now, in the case of \((I_g)\):
\[
\sem{\zxfig{gn-1-1-zero}}^{(I_g)}
=\begin{pmatrix}1&0\\0&0\end{pmatrix}
\ne
\begin{pmatrix}1&0\\0&1\end{pmatrix}
=\sem{\zxfig{single-line}}^{(I_g)}
\]
So, \((I_g)\) is indeed necessary!
\tombstone

\section{Another Optimized Ruleset \texorpdfstring{\(\zxoptprime\)}{ZX\_opt'}}
\label{sec:another-optimized-ruleset}

Now, we introduce a second ruleset adapted from Vilmart~\cite{8785765}, \(\zxoptprime\). We write \(\zxvprime\) for Vilmart's original complete ruleset, so
\[
\zxoptprime=\zxvprime\setminus\{(I_r)\}.
\]
We use the same generators as in \hyperlink{tab:zx-generators}{Figure~1}, keep the usual interpretation, and use only the connectivity meta-rule of Definition~\ref{def:connectivity-meta-rule}, as before.

\begin{center}
\begingroup
\setlength{\tabcolsep}{4pt}
\renewcommand{\arraystretch}{1.25}
\resizebox{0.725\textwidth}{!}{%
\begin{tabular}{|c|}
\hline
\strut\\
\hypertarget{rule:zxoptprime-S}{\sOneTableDiagram}\qquad\quad~\hypertarget{rule:zxoptprime-Ig}{\begin{tikzpicture}
	\begin{pgfonlayer}{nodelayer}
		\node [style=gn] (0) at (-0.75, 0) {};
		\node [style=none] (1) at (0, 0) {=};
		\node [style=none] (2) at (-0.75, 0.75) {};
		\node [style=none] (3) at (-0.75, -0.75) {};
		\node [style=none] (4) at (0.5, 0.75) {};
		\node [style=none] (5) at (0.5, -0.75) {};
		\node [style=none] (6) at (0, -0.25) {(I$_g$)};
	\end{pgfonlayer}
	\begin{pgfonlayer}{edgelayer}
		\draw (2.center) to (0);
		\draw (0) to (3.center);
		\draw (4.center) to (5.center);
	\end{pgfonlayer}
\end{tikzpicture}
}\qquad\quad~\hypertarget{rule:zxoptprime-IV}{\begin{tikzpicture}
	\begin{pgfonlayer}{nodelayer}
		\node [style=gn] (0) at (-0.75, -0.25) {};
		\node [style=gn] (1) at (-1.25, -0.25) {$\alpha$};
		\node [style=none] (2) at (0, 0) {=};
		\node [style=rn] (3) at (-0.75, 0.25) {};
		\node [style=rn] (4) at (-1.25, 0.25) {};
		\node [style=none] (5) at (0, -0.25) {(IV)};
		\node [style=none] (6) at (0.5, 0.25) {};
		\node [style=none] (7) at (0.5, -0.25) {};
		\node [style=none] (8) at (1, 0.25) {};
		\node [style=none] (9) at (1, -0.25) {};
	\end{pgfonlayer}
	\begin{pgfonlayer}{edgelayer}
		\draw [style=none, bend left=45, looseness=1.25] (0) to (3);
		\draw [style=none, bend right=45, looseness=1.25] (0) to (3);
		\draw [style=none] (3) to (0);
		\draw [style=none] (1) to (4);
		\draw [style=dashed] (6.center) to (7.center);
		\draw [style=dashed] (9.center) to (7.center);
		\draw [style=dashed] (6.center) to (8.center);
		\draw [style=dashed] (8.center) to (9.center);
	\end{pgfonlayer}
\end{tikzpicture}}\\[1.25em]
\hypertarget{rule:zxoptprime-CP}{\begin{tikzpicture}
	\begin{pgfonlayer}{nodelayer}
		\node [style=gn] (0) at (0.75, 0) {};
		\node [style=none] (1) at (2.25, -0.25) {};
		\node [style=none] (2) at (0.5, -0.5) {};
		\node [style=rn] (3) at (2.25, 0.25) {};
		\node [style=none] (4) at (1, -0.5) {};
		\node [style=rn] (5) at (0.75, 0.5) {};
		\node [style=rn] (6) at (2.75, 0.25) {};
		\node [style=none] (7) at (2.75, -0.25) {};
		\node [style=none] (8) at (1.5, 0) {$=$};
		\node [style=rn] (9) at (0, 0.25) {};
		\node [style=gn] (10) at (0, -0.25) {};
		\node [style=none] (11) at (1.5, -0.25) {(CP)};
	\end{pgfonlayer}
	\begin{pgfonlayer}{edgelayer}
		\draw [style=none] (5) to (0);
		\draw [style=none, bend right=23, looseness=1.00] (0) to (2.center);
		\draw [style=none, bend left=23, looseness=1.00] (0) to (4.center);
		\draw [style=none] (3) to (1.center);
		\draw [style=none] (6) to (7.center);
		\draw (9) to (10);
	\end{pgfonlayer}
\end{tikzpicture}}\qquad\qquad\hypertarget{rule:zxoptprime-B}{\begin{tikzpicture}
	\begin{pgfonlayer}{nodelayer}
		\node [style=none] (0) at (1.25, 0.75) {};
		\node [style=rn] (1) at (-1.5, -0.25) {};
		\node [style=none] (2) at (0.75, -0.75) {};
		\node [style=none] (3) at (-0.75, 1) {};
		\node [style=none] (4) at (0.75, 0.75) {};
		\node [style=none] (5) at (-1.5, -0.75) {};
		\node [style=none] (6) at (-1.5, 1) {};
		\node [style=gn] (7) at (-0.75, 0.5) {};
		\node [style=none] (8) at (0, -0) {$=$};
		\node [style=rn] (9) at (-0.75, -0.25) {};
		\node [style=gn] (10) at (1, -0.25) {};
		\node [style=gn] (11) at (-1.5, 0.5) {};
		\node [style=none] (12) at (-0.75, -0.75) {};
		\node [style=none] (13) at (1.25, -0.75) {};
		\node [style=rn] (14) at (1, 0.25) {};
		\node [style=rn] (15) at (-2, 0.25) {};
		\node [style=gn] (16) at (-2, -0.25) {};
		\node [style=none] (17) at (0, -0.25) {(B)};
	\end{pgfonlayer}
	\begin{pgfonlayer}{edgelayer}
		\draw [style=none] (12.center) to (9);
		\draw [style=none] (5.center) to (1);
		\draw [style=none] (7) to (3.center);
		\draw [style=none, bend right=23, looseness=1.00] (9) to (7);
		\draw [style=none] (11) to (6.center);
		\draw [style=none, bend left=23, looseness=1.00] (1) to (11);
		\draw [style=none, bend right=23, looseness=1.00] (13.center) to (10);
		\draw [style=none] (10) to (14);
		\draw [style=none, bend left=23, looseness=1.00] (14) to (4.center);
		\draw [style=none, bend right=23, looseness=1.00] (14) to (0.center);
		\draw [bend right=23, looseness=1.00] (10) to (2.center);
		\draw (11) to (9);
		\draw (7) to (1);
		\draw (15) to (16);
	\end{pgfonlayer}
\end{tikzpicture}}\\[1.25em]
\hypertarget{rule:zxoptprime-H}{\begin{tikzpicture}
	\begin{pgfonlayer}{nodelayer}
		\node [style=none] (0) at (0.7500001, 1) {};
		\node [style=none] (1) at (-0.7500001, -1) {};
		\node [style=H box] (2) at (-0.7500001, 0.5000001) {};
		\node [style=rn] (3) at (-1.25, -0) {\footnotesize$~\alpha~$};
		\node [style=H box] (4) at (-1.75, -0.5000001) {};
		\node [style=none] (5) at (-0.7500001, 1) {};
		\node [style=H box] (6) at (-0.7500001, -0.5000001) {};
		\node [style=none] (7) at (-1.25, 0.7500001) {$\cdots$};
		\node [style=none] (8) at (0, -0) {$=$};
		\node [style=gn] (9) at (1.25, -0) {\footnotesize$~\alpha~$};
		\node [style=none] (10) at (1.75, -1) {};
		\node [style=none] (11) at (1.75, 1) {};
		\node [style=none] (12) at (-1.75, 1) {};
		\node [style=H box] (13) at (-1.75, 0.5000001) {};
		\node [style=none] (14) at (-1.75, -1) {};
		\node [style=none] (15) at (0.7500001, -1) {};
		\node [style=none] (16) at (-1.25, -0.7500001) {$\cdots$};
		\node [style=none] (17) at (1.25, 0.7500001) {$\cdots$};
		\node [style=none] (18) at (1.25, -0.7500001) {$\cdots$};
		\node [style=none] (19) at (0, -0.25) {(H)};
	\end{pgfonlayer}
	\begin{pgfonlayer}{edgelayer}
		\draw [bend right, looseness=1.00] (3) to (4);
		\draw [bend left, looseness=1.00] (3) to (6);
		\draw (6) to (1.center);
		\draw (4) to (14.center);
		\draw [bend left, looseness=1.00] (3) to (13);
		\draw [bend right, looseness=1.00] (3) to (2);
		\draw (2) to (5.center);
		\draw (13) to (12.center);
		\draw [bend left=23, looseness=1.00] (9) to (0.center);
		\draw [bend right=23, looseness=1.00] (9) to (11.center);
		\draw [bend right=23, looseness=1.00] (9) to (15.center);
		\draw [bend left=23, looseness=1.00] (9) to (10.center);
	\end{pgfonlayer}
\end{tikzpicture}}\qquad\qquad\hypertarget{rule:zxoptprime-EUprime}{\begin{tikzpicture}
	\begin{pgfonlayer}{nodelayer}
		\node [style=none] (0) at (0, 0) {=};
		\node [style=gn] (1) at (0.75, 0) {$\beta_2$};
		\node [style=rn] (2) at (0.75, 0.55) {$\beta_1$};
		\node [style=rn] (3) at (0.75, -0.55) {$\beta_3$};
		\node [style=gn] (5) at (-0.75, 0.5) {$\alpha_1$};
		\node [style=gn] (6) at (-0.75, -0.5) {$\alpha_2$};
		\node [style=none] (7) at (0.75, 1) {};
		\node [style=none] (8) at (0.75, -1) {};
		\node [style=none] (9) at (-0.75, 1) {};
		\node [style=none] (10) at (-0.75, -1) {};
		\node [style=rn] (11) at (1.5, -0.25) {$\pi$};
		\node [style=gn] (12) at (1.5, -0.75) {$\gamma$};
		\node [style=rn] (13) at (1.5, 0.75) {};
		\node [style=gn] (14) at (1.5, 0.25) {};
		\node [style=none] (15) at (0, -0.25) {\textnormal{(EU')}};
		\node [style={{H box}}] (16) at (-0.75, 0) {};
	\end{pgfonlayer}
	\begin{pgfonlayer}{edgelayer}
		\draw [style=none] (7.center) to (8.center);
		\draw [style=none] (9.center) to (10.center);
		\draw [style=none] (12) to (11);
		\draw [style=none, bend left, looseness=1.25] (14) to (13);
		\draw [style=none, bend right, looseness=1.25] (14) to (13);
		\draw [style=none] (13) to (14);
	\end{pgfonlayer}
\end{tikzpicture}}\\
\strut\\
\hline
\end{tabular}}

\vspace{0.5em}

\begin{minipage}{\textwidth}
\small
\textbf{Figure 3.} Set of rules \zxoptprime for the ZX calculus with scalars. The above labels and rewrites are lightly adapted from Vilmart~\cite{8785765}. The right-hand side of (IV) is an empty diagram. (...) denote zero or more wires, while (\rotatebox{45}{\raisebox{-0.4em}{$\cdots$}}) denote one or more wires. In rule (EU'), $\beta_1,\beta_2,\beta_3$ and $\gamma$ can be determined as follows: $x^+:=\frac{\alpha_1+\alpha_2}{2}$, $x^-:=x^+-\alpha_2$, $z := -\sin{x^+}+i\cos{x^-}$ and $z' := \cos{x^+}-i\sin{x^-}$, then $\beta_1 = \arg z + \arg z'$, $\beta_2 = 2\arg\left(i+\left|\frac{z}{z'}\right|\right)$, $\beta_3 = \arg z - \arg z'$, $\gamma = x^+-\arg(z)+\frac{\pi-\beta_2}{2}$ where by convention $\arg(0):=0$ and $z'=0\implies \beta_2=0$.
\end{minipage}
\endgroup
\end{center}

The necessity arguments of Vilmart~\cite{8785765} apply unchanged to \((S)\), \((H)\), and \((CP)\), since the replacement rules \((IV)\) and \((EU')\) involve only nullary and \(1\to1\) structures. Moreover, \((IV)\) is the only rule in \(\zxoptprime\) with an empty side, and \((EU')\) is the only non-linear rule. Thus, \((S)\), \((IV)\), \((H)\), \((EU')\), and \((CP)\) are all necessary. We consider the remaining rules \((B)\) and \((I_g)\) below.

\subsection{\texorpdfstring{\((I_r)\)}{(Ir)} is derivable in \texorpdfstring{\(\zxoptprime\)}{ZXopt'}}
\label{sec:zxoptprime-ir-rule-derivable}

As with the other ruleset, \((I_r)\) is derivable.

The following derivation of \hyperref[lem:zxoptprime-ir-rule-derivable]{\((I_r)\)} uses only axioms from the \(\mathrm{ZX}_{\mathrm{opt}}'\) rule set, which has some overlap with \(\mathrm{ZX}_{\mathrm{opt}}\), hence we are allowed to use the \hyperref[lem:ir-derivation-L1]{\((H_0)\)} and \hyperref[lem:ir-derivation-L2]{\((HC)\)} lemmas proved in \hyperref[app:ir-rule-derivation]{Appendix A}.

It is worth noting that this derivation for \hyperref[lem:zxoptprime-ir-rule-derivable]{\((I_r)\)} uses a different strategy than the other ruleset, where the primary concern was scalar management.  Here, scalars take a back-seat, since \hyperlink{rule:zxoptprime-IV}{\((IV)\)} is more flexible.  Rather, we focus on deriving the self-inverse property of the Hadamard through some relatively straightforward lemmas regarding cancellation and commutation on \(1 \to 1\) arity structures. In fact, every lemma used in this derivation has \(1 \to 1\) arity except for \hyperref[eq:prime-ir-Fig3Hx]{\((HX)\)}.

\begin{theorem}
\label{lem:zxoptprime-ir-rule-derivable}
The \((I_r)\) rule is derivable.
\begin{equation*}
\zxoptprime\vdash
\vcenter{\hbox{\primeirtikzinput{fig3-ir-theorem-statement.tikz}}}
\end{equation*}
\end{theorem}

\noindent\textbf{Proof.}
\begin{center}
\makebox[\linewidth][c]{\primeirtikzinput{ir-final-proof.tikz}}
\end{center}
\tombstone

See \hyperref[app:ir-prime-rule-derivation]{Appendix C} for details.

\subsection{On the Necessity of the Bialgebra \texorpdfstring{\((B)\)}{(B)} rule in \texorpdfstring{\(\zxoptprime\)}{ZXopt'}}
\label{sec:zxoptprime-necessity-bialgebra-b-rule}

We can simply re-use the exact same counter-model from the \hyperref[sec:necessity-bialgebra-b-rule]{\(\zxopt\) necessity of \((B)\)}. It suffices to show the soundness of the new rules \((IV)\) and \((EU')\), since we know that all the other non-target rules hold, and \((B)\) fails as before.

\begin{theorem}
\label{lem:zxoptprime-b-rule-necessity}
The \((B)\) rule is necessary:
\[
\zxoptprime\setminus\{(B)\}\nvdash (B).
\]
\end{theorem}

\noindent\textbf{Proof.}
Consider the interpretation \(\sem{-}^{(B)}\) defined in \hyperref[sec:defining-interpretation-functor]{Section~\ref*{sec:defining-interpretation-functor}}. It suffices to show that \((IV)\) and \((EU')\) hold under this interpretation.
\begin{itemize}
\item \((IV)\):
\[
\sem{\zxfig{IV-LHS-only}}^{(B)}
=1=
\sem{\zxfig{empty-diagram}}^{(B)}.
\]
On the left-hand side, for any \(\alpha\), the two disconnected scalars evaluate to \(1/(2\sqrt2)\) and \(2\sqrt2\), respectively, giving a product of \(1\).  The right-hand side evaluates trivially to \(1\).
\item \((EU')\):
\[
\sem{\zxfig{EU-prime-LHS-only}}^{(B)} = D_{\alpha_2}H'D_{\alpha_1} = e^{3i\gamma}H'D_{\beta_3}H'D_{\beta_2}H'D_{\beta_1}H' = \sem{\zxfig{EU-prime-RHS-only}}^{(B)}.
\]
where \(D_\theta=(Z')^\theta_{1,1}\).  Just like before in the \((EU)\) check, the two scalars on the right-hand side contribute \(e^{3i\gamma}\) in total.  Then, tensoring the corresponding one-qubit Euler identity three times and conjugating by \(O\), which commutes with every \(D_\theta\), we find agreement on the final displayed equality.
\end{itemize}
\[
\therefore\ (B)\text{ is necessary.}
\qquad\blacksquare
\]

\subsection{On the Necessity of the Green Identity \texorpdfstring{\((I_g)\)}{(Ig)} rule in \texorpdfstring{\(\zxoptprime\)}{ZXopt'}}
\label{sec:zxoptprime-necessity-green-identity-ig-rule}

\begin{theorem}
\label{lem:zxoptprime-ig-rule-necessity}
The \((I_g)\) rule is necessary:
\[
\zxoptprime\setminus\{(I_g)\}\nvdash (I_g).
\]
\end{theorem}

\noindent\textbf{Proof.}
Consider the interpretation \(\sem{-}^{(I_g)}\) defined in \hyperref[sec:necessity-green-identity-ig-rule]{Section~\ref*{sec:necessity-green-identity-ig-rule}}. As before, it suffices to show that \((IV)\) and \((EU')\) hold under this interpretation.
\begin{itemize}
\item \((IV)\):
\[
\sem{\zxfig{IV-LHS-only}}^{(I_g)}
=1=
\sem{\zxfig{empty-diagram}}^{(I_g)}.
\]
\item \((EU')\):
\[
\sem{\zxfig{EU-prime-LHS-only}}^{(I_g)} = P_{1,1} = \sem{\zxfig{EU-prime-RHS-only}}^{(I_g)}.
\]
\end{itemize}
\[
\therefore\ (I_g)\text{ is necessary.}
\]
\tombstone

So, \(\zxoptprime\) is another complete, minimal ruleset.

\section*{Acknowledgements}

We thank Miriam Backens for thoughtful discussions about this work, for checking the computations, and for providing deeper insights into the problem.

\bibliographystyle{eptcs}
\bibliography{references}

\appendix
\section{Derivations for the Red Identity Rule}
\label{app:ir-rule-derivation}

Throughout this appendix, we work in
\[
\zxopt=\zxv\setminus\{(I_r)\}
=\{(S),(I_g),(E),(CP),(B),(EU),(HD),(H)\}
+\text{ connectivity}.
\]

The following derivation of \((I_r)\) is highly technical, and most of the difficulty comes from the critical tracking of scalars and the scalar cleanup rule in \hyperref[eq:ir-derivation-L7]{Lemma~\ref*{lem:ir-derivation-L7}}.

We first derive two useful identities from the ruleset:

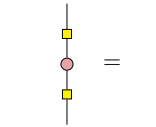
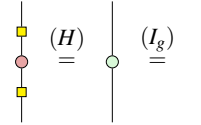
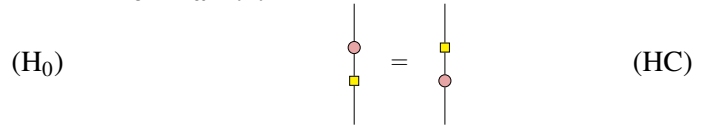
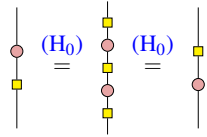
\newcommand{\irshortstatement}[3]{%
\begin{equation*}
\vcenter{\hbox{\hkstikzitinput{HKS/figures/ir-derivation/renaud-short/#1}}}
\tag{\ensuremath{#3}}
\label{eq:ir-derivation-#2}
\end{equation*}%
}
\newcommand{\irshortrow}[1]{%
\leavevmode\par
\nopagebreak[4]
\begin{center}
\makebox[\linewidth][c]{#1}%
\end{center}
}
\newcommand{\irshortjoin}{\hspace{1em}}
\newcommand{\scalarpairnodes}[3]{%
\begin{scope}[shift={(#2,#3)}]
	\node [style=X dot] (#1-chi-x) at (-0.6,0.35) {};
	\node [style=Z dot] (#1-chi-z) at (-0.6,-0.35) {};
	\node [style=X phase dot] (#1-omega-x) at (0.6,0.45) {$\pi$};
	\node [style=Z dot] (#1-omega-z) at (0.6,-0.45) {};
\end{scope}%
}
\newcommand{\scalarpairedges}[1]{%
\draw (#1-chi-x) to (#1-chi-z);
\draw [bend left=45] (#1-chi-x) to (#1-chi-z);
\draw [bend right=45] (#1-chi-x) to (#1-chi-z);
\draw (#1-omega-x) to (#1-omega-z);
}

\Needspace{0.30\textheight}
\noindent
\begin{minipage}[t]{0.48\linewidth}
\vspace{0pt}
\begin{lemma}
\label{lem:ir-derivation-L1}
\irshortstatement{ir-l01-statement.tikz}{L1}{\mathrm{H}_0}
\end{lemma}
\begin{irproof}
\widehkstikzitinput{HKS/figures/ir-derivation/renaud-short/ir-l01-proof.tikz}
\vspace{-2.6\baselineskip}
\end{irproof}
\end{minipage}\hfill
\begin{minipage}[t]{0.48\linewidth}
\vspace{0pt}
\begin{lemma}
\label{lem:ir-derivation-L2}
\irshortstatement{ir-l02-statement.tikz}{L2}{\mathrm{HC}}
\end{lemma}
\begin{irproof}
\widehkstikzitinput{HKS/figures/ir-derivation/renaud-short/ir-l02-proof.tikz}
\end{irproof}
\end{minipage}

\vspace{\baselineskip}

Now, we derive a rule which lets us remove the red identity along with the two scalar factors shown:

\Needspace{0.35\textheight}
\begin{lemma}
\label{lem:ir-derivation-L3}
\irshortstatement{ir-atomic-l03-statement.tikz}{L3}{\mathrm{LI}_r}
\end{lemma}
\begin{irproof}
\irshortrow{%
	\hkstikzitinput{HKS/figures/ir-derivation/renaud-short/ir-atomic-l03-proof-a.tikz}\irshortjoin
	\hkstikzitinput{HKS/figures/ir-derivation/renaud-short/ir-atomic-l03-proof-b.tikz}\irshortjoin
	\hkstikzitinput{HKS/figures/ir-derivation/renaud-short/ir-atomic-l03-proof-b2.tikz}%
}
\end{irproof}

The Hadamard is self-inverse up to a scalar residue:

\Needspace{0.25\textheight}
\begin{lemma}
\label{lem:ir-derivation-L4}
\irshortstatement{ir-atomic-l04-statement.tikz}{L4}{\mathrm{H}^{-1}}
\end{lemma}
\begin{irproof}
\widehkstikzitinput{HKS/figures/ir-derivation/renaud-short/ir-atomic-l04-proof.tikz}
\end{irproof}

We derive a scalar-tracked version of red spider fusion:

\Needspace{0.35\textheight}
\begin{lemma}
\label{lem:ir-derivation-L5}
\irshortstatement{ir-atomic-l07-statement.tikz}{L5}{\mathrm{LS}_r}
\end{lemma}
\begin{irproof}
\widehkstikzitinput{HKS/figures/ir-derivation/renaud-short/ir-atomic-l07-proof-a.tikz}
\vspace{-0.75\baselineskip}
\irshortrow{%
	\hkstikzitinput{HKS/figures/ir-derivation/renaud-short/ir-atomic-l07-proof-a2.tikz}\irshortjoin
	\hkstikzitinput{HKS/figures/ir-derivation/renaud-short/ir-atomic-l07-proof-b.tikz}%
}
\vspace{-1.5\baselineskip}
\end{irproof}

\Needspace{0.45\textheight}
We obtain a useful scalar-tracked form of the familiar Hopf law:

\begin{lemma}
\label{lem:ir-derivation-L6}
\irshortstatement{ir-renaud-l06-statement.tikz}{L6}{\mathrm{HF}}
\end{lemma}
\begin{irproof}
\irshortrow{%
	\hkstikzitinput{HKS/figures/ir-derivation/renaud-short/ir-renaud-l06-proof-1.tikz}\irshortjoin
	\hkstikzitinput{HKS/figures/ir-derivation/renaud-short/ir-renaud-l06-proof-2.tikz}%
}
\irshortrow{%
	\hkstikzitinput{HKS/figures/ir-derivation/renaud-short/ir-renaud-l06-proof-3.tikz}\irshortjoin
	\hkstikzitinput{HKS/figures/ir-derivation/renaud-short/ir-renaud-l06-proof-a2.tikz}%
}
\irshortrow{%
	\hkstikzitinput{HKS/figures/ir-derivation/renaud-short/ir-renaud-l06-proof-a3.tikz}\irshortjoin
	\hkstikzitinput{HKS/figures/ir-derivation/renaud-short/ir-renaud-l06-proof-6.tikz}\irshortjoin
	\hkstikzitinput{HKS/figures/ir-derivation/renaud-short/ir-renaud-l06-proof-7.tikz}%
}
\end{irproof}

We now obtain the scalar cleanup rule which lets us remove the following pair of scalar factors:

\Needspace{0.45\textheight}
\begin{lemma}
\label{lem:ir-derivation-L7}
\irshortstatement{ir-atomic-l14-statement.tikz}{L7}{\mathrm{IV}'}
\end{lemma}
\begin{irproof}
\irshortrow{%
	\hkstikzitinput{HKS/figures/ir-derivation/renaud-short/ir-renaud-l07-proof-1.tikz}\irshortjoin
	\hkstikzitinput{HKS/figures/ir-derivation/renaud-short/ir-renaud-l07-proof-3.tikz}\irshortjoin
	\hkstikzitinput{HKS/figures/ir-derivation/renaud-short/ir-renaud-l07-proof-4.tikz}%
}
\irshortrow{%
	\hkstikzitinput{HKS/figures/ir-derivation/renaud-short/ir-renaud-l07-proof-5.tikz}\irshortjoin
	\hkstikzitinput{HKS/figures/ir-derivation/renaud-short/ir-renaud-l07-proof-6.tikz}\irshortjoin
	\hkstikzitinput{HKS/figures/ir-derivation/renaud-short/ir-renaud-l07-proof-6a.tikz}%
}
\irshortrow{%
	\hkstikzitinput{HKS/figures/ir-derivation/renaud-short/ir-renaud-l07-proof-7.tikz}\irshortjoin
	\hkstikzitinput{HKS/figures/ir-derivation/renaud-short/ir-renaud-l07-proof-7a.tikz}\irshortjoin
	\hkstikzitinput{HKS/figures/ir-derivation/renaud-short/ir-renaud-l07-proof-8.tikz}\irshortjoin
	\hkstikzitinput{HKS/figures/ir-derivation/renaud-short/ir-renaud-l07-proof-9.tikz}%
}
\end{irproof}

Two necessary technical lemmas:

\Needspace{0.45\textheight}
\begin{lemma}
\label{lem:ir-derivation-L8}
\irshortstatement{ir-renaud-l08-statement.tikz}{L8}{\mathrm{IV}''}
\end{lemma}
\begin{irproof}
\irshortrow{%
	\hkstikzitinput{HKS/figures/ir-derivation/renaud-short/ir-renaud-l08-proof-1.tikz}\irshortjoin
	\hkstikzitinput{HKS/figures/ir-derivation/renaud-short/ir-renaud-l08-proof-2.tikz}\irshortjoin
	\hkstikzitinput{HKS/figures/ir-derivation/renaud-short/ir-renaud-l08-proof-3.tikz}\irshortjoin
	\hkstikzitinput{HKS/figures/ir-derivation/renaud-short/ir-renaud-l08-proof-4.tikz}%
}
\irshortrow{%
	\hkstikzitinput{HKS/figures/ir-derivation/renaud-short/ir-renaud-l08-proof-5.tikz}\irshortjoin
	\hkstikzitinput{HKS/figures/ir-derivation/renaud-short/ir-renaud-l08-proof-6.tikz}\irshortjoin
	\hkstikzitinput{HKS/figures/ir-derivation/renaud-short/ir-renaud-l08-proof-7.tikz}%
}
\irshortrow{%
	\hkstikzitinput{HKS/figures/ir-derivation/renaud-short/ir-renaud-l08-proof-7a.tikz}\irshortjoin
	\hkstikzitinput{HKS/figures/ir-derivation/renaud-short/ir-renaud-l08-proof-8a.tikz}\irshortjoin
	\hkstikzitinput{HKS/figures/ir-derivation/renaud-short/ir-renaud-l08-proof-9.tikz}%
}
\end{irproof}

\Needspace{0.40\textheight}
The following lemma is a useful identity allowing the red state to delete a green phase gate of \(-\pi/2\):

\begin{lemma}
\label{lem:ir-derivation-L9}
\irshortstatement{ir-renaud-l09-statement.tikz}{L9}{\mathrm{D}_{-}}
\end{lemma}
\begin{irproof}
\irshortrow{%
	\hkstikzitinput{HKS/figures/ir-derivation/renaud-short/ir-renaud-l09-proof-a.tikz}\irshortjoin
	\hkstikzitinput{HKS/figures/ir-derivation/renaud-short/ir-renaud-l09-proof-b.tikz}%
}
\irshortrow{%
	\hkstikzitinput{HKS/figures/ir-derivation/renaud-short/ir-renaud-l09-proof-b2.tikz}\irshortjoin
	\hkstikzitinput{HKS/figures/ir-derivation/renaud-short/ir-renaud-l09-proof-b3.tikz}\irshortjoin
	\hkstikzitinput{HKS/figures/ir-derivation/renaud-short/ir-renaud-l09-proof-b4.tikz}\irshortjoin
	\hkstikzitinput{HKS/figures/ir-derivation/renaud-short/ir-renaud-l09-proof-c.tikz}%
}
\end{irproof}

A useful identity is derived:

\Needspace{0.35\textheight}
\begin{lemma}
\label{lem:ir-derivation-L10}
\irshortstatement{ir-renaud-l10-statement.tikz}{L10}{\mathrm{D}'_{\pi}}
\end{lemma}
\begin{irproof}
\leavevmode\par
\nopagebreak[4]
\vspace{14pt}
\begin{center}
\makebox[\linewidth][c]{%
	\hkstikzitinput{HKS/figures/ir-derivation/renaud-short/ir-renaud-l10-proof-short.tikz}%
}
\end{center}
\vspace{16.368pt}
\end{irproof}

\Needspace{0.32\textheight}
We reproduce the derivation of \((I_r)\) below for readability.

\begin{theorem}
The \((I_r)\) rule is derivable:
\begin{equation*}
\zxopt\vdash
\vcenter{\hbox{\hkstikzitinput{HKS/figures/ir-derivation/ir-atomic-ir-statement.tikz}}}
\end{equation*}
\end{theorem}

\begin{irproof}
\leavevmode\par
\begin{center}
\makebox[\linewidth][c]{%
	\hkstikzitinput{HKS/figures/ir-derivation/renaud-short/ir-renaud-ir-proof-a.tikz}\hspace{1em}%
	\hkstikzitinput{HKS/figures/ir-derivation/renaud-short/ir-renaud-ir-proof-b.tikz}\hspace{1em}%
	\hkstikzitinput{HKS/figures/ir-derivation/renaud-short/ir-renaud-ir-proof-c.tikz}%
}
\end{center}
\end{irproof}

\section{Necessity of the remaining rules in \texorpdfstring{\(\zxopt\)}{ZXopt}}
\phantomsection\label{app:b-rule-soundness-checks}

\noindent\textbf{Theorem~\ref{thm:b-rule-necessity}.}
\[
\zxopt\setminus\{(B)\}\nvdash (B).
\]

\noindent\textbf{Proof.}
We use the interpretation \(\sem{-}^{(B)}\) from Section~\ref{sec:necessity-bialgebra-b-rule}.
The rule \((B)\) fails, as shown above.
It remains to check soundness of the other rules.

\begin{itemize}
\item \((I_g)\)
\[
\sem{\zxfig{I-g-LHS-only}}^{(B)}=(Z')^0_{1,1}
=I
=\sem{\zxfig{single-line}}^{(B)}.
\]
\[
\therefore\ (I_g)\text{ holds.}
\]

\item \((E)\)
For \(x\in\{0,1\}^3\), let
\[
p_x=e^{i\lvert x\rvert\pi/4},
\qquad
q_x=e^{-i\lvert x\rvert\pi/4}.
\]
Since \(p\) and \(q\) are constant on the weight-one subspace, we have
\[
Np=0,
\qquad
q^TN=0.
\]
Therefore \(Op=p\), \(q^TO^{-1}=q^T\), and
\[
\sem{\zxfig{E-LHS-only}}^{(B)}
=q^TH'p
=q^TH_0p
=\left(
\begin{pmatrix}1&e^{-i\pi/4}\end{pmatrix}
h
\begin{pmatrix}1\\e^{i\pi/4}\end{pmatrix}
\right)^3
=1
=\sem{\zxfig{empty-diagram}}^{(B)}.
\]

\item \((H)\)
\[
\sem{\zxfig{H-LHS-only}}^{(B)}
=(X')^\alpha_{n,m}
=
(H')^{\otimes m}
\circ
(Z')^\alpha_{n,m}
\circ
(H')^{\otimes n}
=\sem{\hAlphaRhsDiagram}^{(B)}.
\]
By definition of the X-spiders.

\item \((CP)\)
If we denote \(u:=\sum_{x\in\{0,1\}^3}\lvert x\rangle\), then we have that \(Nu=0=N\lvert000\rangle\), so \(O\) fixes both \(u\) and \(\lvert000\rangle\). Thus,
\[
(X')^0_{0,1}=H'u=O^{-1}H_0Ou=2\sqrt2\,\lvert000\rangle.
\]
Then, both the scalar and two-output components on the left therefore evaluate to \(2\sqrt2\) and\linebreak \(2\sqrt2\,\lvert000\rangle^{\otimes2}\), respectively. On the right-hand side, we evaluate the tensor square of the same red state. So both sides equal \(8\lvert000\rangle^{\otimes2}\).
\[
\sem{\zxbig{CP-LHS-only}}^{(B)}
=8\lvert 000\rangle\otimes\lvert 000\rangle
=\sem{\zxbig{CP-RHS-only}}^{(B)}.
\]

\item \((S)\) Spider fusion holds by construction.  The \(k\geq1\) internal connecting wires force to the same basis vector \(\lvert x\rangle\), and
\[
e^{i\lvert x\rvert\alpha}e^{i\lvert x\rvert\beta}
=
e^{i\lvert x\rvert(\alpha+\beta)}.
\]
Therefore
\[
\sem{\sLhsDiagram}^{(B)}
=
\sem{\sRhsDiagram}^{(B)}.
\]

\item \((EU)\) note:
\[
\sem{\zxsmall{EU-dumbbell-only}}^{(B)}
=\frac{1}{2\sqrt2}
\quad \& \quad
\sem{\zxsmall{EU-bottom-scalar-only}}^{(B)}
=2\sqrt2\,e^{3i\gamma},
\]
so scalars contribute \(e^{3i\gamma}\).

We define \(d_\theta\) and \(D_\theta\) as shorthand, and recall the definitions of \(H_0\) and \(H'\):
\[
d_\theta:=\operatorname{diag}(1,e^{i\theta}),
\qquad
D_\theta:=(Z')^\theta_{1,1}=d_\theta^{\otimes 3},
\qquad
H_0=h^{\otimes 3},
\qquad
H'=O^{-1}H_0O.
\]
Since \(N\) only acts within the weight-one subspace, \(O\) commutes with every \(D_\theta\).
\[
\sem{\zxsmall{EU-LHS-only}}^{(B)}
=(Z')^{\alpha_3}_{1,1}\circ (X')^{\alpha_2}_{1,1}\circ (Z')^{\alpha_1}_{1,1}
=D_{\alpha_3}H'D_{\alpha_2}H'D_{\alpha_1}.
\]
\[
\sem{\zxsmall{EU-RHS-only}}^{(B)}
=e^{3i\gamma}\cdot (X')^{\beta_3}_{1,1}\circ (Z')^{\beta_2}_{1,1}\circ (X')^{\beta_1}_{1,1}
=e^{3i\gamma}\cdot H'D_{\beta_3}H'D_{\beta_2}H'D_{\beta_1}H'.
\]

The corresponding scalar-free qubit Euler identity is
\[
 d_{\alpha_3}h d_{\alpha_2}h d_{\alpha_1}
=e^{i\gamma}h d_{\beta_3}h d_{\beta_2}h d_{\beta_1}h.
\]
Tensoring this identity three times gives
\[
D_{\alpha_3}H_0D_{\alpha_2}H_0D_{\alpha_1}
=
e^{3i\gamma}H_0D_{\beta_3}H_0D_{\beta_2}H_0D_{\beta_1}H_0.
\]
Conjugating by \(O\) gives
\[
D_{\alpha_3}H'D_{\alpha_2}H'D_{\alpha_1}
=
e^{3i\gamma}H'D_{\beta_3}H'D_{\beta_2}H'D_{\beta_1}H',
\]
as desired!

\item \((HD)\)
We write as shorthand:
\[
D_{\pi/2}=d_{\pi/2}^{\otimes 3},
\qquad
t(-\pi/2)=(1,-i,-i,-1,-i,-1,-1,i)^T.
\]
And compute:
\[
H't(-\pi/2)
=
e^{-3i\pi/4}(1,i,i,-1,i,-1,-1,-i)^T.
\]
\[
\sem{\zxfig{HD-RHS-only}}^{(B)}
=D_{\pi/2}H'\operatorname{diag}\bigl(H't(-\pi/2)\bigr)H'D_{\pi/2}
=H'
=\sem{\zxfig{Hadamard}}^{(B)}.
\]
\end{itemize}
\tombstone

\phantomsection\label{app:ig-rule-soundness-checks}

\bigskip

\noindent\textbf{Theorem~\ref{thm:ig-rule-necessity}.}
\[
\zxopt\setminus\{(I_g)\}\nvdash (I_g).
\]

\noindent\textbf{Proof.}
We use the interpretation \(\sem{-}^{(I_g)}\) from Section~\ref{sec:necessity-green-identity-ig-rule}.
The rule \((I_g)\) fails, as shown above.
It remains to check soundness of the other rules.

Since all scalars interpret as \(1\), we suppress them throughout.

\begin{itemize}
\item \((S)\)
\[
\sem{\sLhsDiagram}^{(I_g)}
=P_{n+n',m+m'}
=\sem{\sRhsDiagram}^{(I_g)}
\]

\item \((I_r)\)
\[
\sem{\zxfig{rn-1-1-zero}}^{(I_g)}
=\begin{pmatrix}1&0\\0&1\end{pmatrix}
=\sem{\zxfig{single-line}}^{(I_g)}
\]

\item \((E)\)
The two spiders in the left-hand side are both interpreted by one-point relations, so
\[
\sem{\zxfig{E-LHS-only}}^{(I_g)}
=P_{1,0}\circ P_{0,1}
=1
=\sem{\zxfig{empty-diagram}}^{(I_g)}.
\]

\item \((CP)\)
\[
\sem{\zxbig{CP-LHS-only}}^{(I_g)}
=P_{1,2}\circ P_{0,1}\stackrel{(!)}{=}P_{0,2}
=(1\ 0\ 0\ 0)
\]
\[
\sem{\zxsmall{rn-0-1}\quad\zxsmall{rn-0-1}}^{(I_g)}
=P_{0,1}\otimes P_{0,1}\stackrel{(!!)}{=}P_{0,2}
=(1\ 0\ 0\ 0)
\]
where we have used the \(P\)-identities:
\[
P_{k,m}\circ P_{n,k}=P_{n,m} \tag{!}
\]
\[
P_{n,m}\otimes P_{n',m'}=P_{n+n',m+m'} \tag{!!}
\]

\item \((B)\)
\[
\sem{\zxbig{B-LHS-only}}^{(I_g)}
=(P_{2,1}\otimes P_{2,1})\circ(1\otimes\sigma\otimes1)\circ(P_{1,2}\otimes P_{1,2})
=\begin{pmatrix}1&0&0&0\\0&0&0&0\\0&0&0&0\\0&0&0&0\end{pmatrix}
\]
\[
\sem{\zxbig{B-RHS-only}}^{(I_g)}
=P_{2,2}\circ P_{2,2}=P_{2,2}
=\begin{pmatrix}1&0&0&0\\0&0&0&0\\0&0&0&0\\0&0&0&0\end{pmatrix}
\]

\item \((HD)\)
\[
\sem{\zxfig{HD-RHS-only}}^{(I_g)}
=P_{1,1}\circ P_{2,1}\circ(P_{1,1}\otimes P_{0,1})
=P_{1,1}\circ P_{2,1}\circ P_{1,2}
=P_{1,1}
=\sem{\zxfig{Hadamard}}^{(I_g)}
\]

\item \((H)\)
When \(n+m=2\) and \(\alpha=0\), we have:
\begin{itemize}
\item \((n=0,\ m=2)\)
\[
P_{0,2}=(1\ 0\ 0\ 0)=(1\ 0\ 0\ 1)
\begin{psmallmatrix}
1&0&0&0\\
0&0&0&0\\
0&0&0&0\\
0&0&0&0
\end{psmallmatrix}
=(P_{1,1}\otimes P_{1,1})\circ C_{0,2}.
\]
\item \((n=1,\ m=1)\)
\[
P_{1,1}
=\begin{psmallmatrix}1&0\\0&0\end{psmallmatrix}
=\begin{psmallmatrix}1&0\\0&0\end{psmallmatrix}
\begin{psmallmatrix}1&0\\0&1\end{psmallmatrix}
\begin{psmallmatrix}1&0\\0&0\end{psmallmatrix}
=P_{1,1}\circ C_{1,1}\circ P_{1,1}.
\]
\item \((n=2,\ m=0)\)
\[
P_{2,0}
=\begin{psmallmatrix}1\\0\\0\\0\end{psmallmatrix}
=\begin{psmallmatrix}
1&0&0&0\\
0&0&0&0\\
0&0&0&0\\
0&0&0&0
\end{psmallmatrix}
\begin{psmallmatrix}1\\0\\0\\1\end{psmallmatrix}
=C_{2,0}\circ(P_{1,1}\otimes P_{1,1}).
\]
\end{itemize}
Otherwise,
\[
\sem{\zxfig{H-LHS-only}}^{(I_g)}
=P_{1,1}^{\otimes m}\circ P_{n,m}\circ P_{1,1}^{\otimes n}
=P_{m,m}\circ P_{n,m}\circ P_{n,n}=P_{n,m}
=\sem{\hAlphaRhsDiagram}^{(I_g)}
\]

\item \((EU)\) Since scalars contribute \(1\), and phases
do not matter:
\[
\sem{\zxsmall{EU-LHS-only}}^{(I_g)}=P_{1,1}
=\sem{\euRhsNoScalarsDiagram}^{(I_g)}
=\sem{\zxsmall{EU-RHS-only}}^{(I_g)}.
\]
\end{itemize}
\tombstone

\section{Derivations for the Red Identity Rule in \texorpdfstring{\(\zxoptprime\)}{ZXopt'}}
\label{app:ir-prime-rule-derivation}

Throughout this appendix, we work in
\[
\zxoptprime
=\{(S),(I_g),(IV),(CP),(B),(H),(EU')\}
+\text{ connectivity}.
\]

\newcommand{\primeirdiagramstatement}[3]{%
\begin{equation*}
\vcenter{\hbox{\primeirtikzinput{#1}}}
\tag{\ensuremath{#2}}
\label{eq:prime-ir-#3}
\end{equation*}%
}
\newcommand{\primeirderivabilitystatement}[3]{%
\begin{equation*}
\zxoptprime\vdash\vcenter{\hbox{\primeirtikzinput{#1}}}
\label{eq:prime-ir-#3}
\end{equation*}%
}
\newcommand{\primeirrow}[1]{\par\noindent\makebox[\linewidth][c]{\primeirtikzinput{#1}}\par}
\newenvironment{primeircompactlemma}[1][]{%
  \par\medskip\refstepcounter{theorem}\noindent\textbf{Lemma~\thetheorem.}%
  \if\relax\detokenize{#1}\relax\else\ #1\fi
  \par\nopagebreak[4]
}{\par}
\newenvironment{primeircompacttheorem}[1][]{%
  \par\medskip\refstepcounter{theorem}\noindent\textbf{Theorem~\thetheorem.}%
  \if\relax\detokenize{#1}\relax\else\ #1\fi
  \par\nopagebreak[4]
}{\par}
\newcommand{\primeircompactproof}{\par\noindent\textbf{Proof.}\par\nopagebreak[4]}

We have access to the \((HC)\) and \((H_0)\) lemmas derived in \hyperref[app:ir-rule-derivation]{Appendix A}, since they rely on shared rules.

Now, we first derive two technical lemmas, \hyperref[eq:prime-ir-Fig3Hx]{\((HX)\)}, which is like a color change rule, and \hyperref[eq:prime-ir-Fig3RT]{\((RT)\)}, which sets us up for applications of the copy \hyperlink{rule:zxoptprime-CP}{\((CP)\)} rule.

\begin{primeircompactlemma}
\primeirdiagramstatement{fig3-ir-hx-statement.tikz}{HX}{Fig3Hx}
\primeircompactproof
\primeirrow{fig3-ir-hx-proof.tikz}
\tombstone
\end{primeircompactlemma}

\begin{primeircompactlemma}
\primeirdiagramstatement{fig3-ir-r-statement.tikz}{RT}{Fig3RT}
\primeircompactproof
\par\noindent\makebox[\linewidth][c]{%
  \primeirtikzinput{rt-11-a.tikz}%
  \hspace{0.15cm}%
  \primeirtikzinput{rt-11-a-tail.tikz}%
}\par
\par\noindent\makebox[\linewidth][c]{%
  \primeirtikzinput{rt-11-b.tikz}%
  \hspace{0.15cm}%
  \primeirtikzinput{rt-11-c.tikz}%
}\par
\tombstone
\end{primeircompactlemma}

We now derive a green phase commutation rule.

\begin{primeircompactlemma}[When \( |\alpha|<\pi \),]
\primeirdiagramstatement{ch2-general-statement.tikz}{ZC}{GeneralCHTwo}
\primeircompactproof
\vspace{1.2cm}
\par\noindent\makebox[\linewidth][c]{%
  \primeirtikzinput{ch2-general-a.tikz}%
  \hspace{0.15cm}%
  \primeirtikzinput{ch2-general-c.tikz}%
}\par
\tombstone
\end{primeircompactlemma}

Here we derive a crucial Clifford+\(T\) phase cancellation rule for red spiders. This is conceptually close to \hyperref[eq:prime-ir-Ir]{\((I_r)\)} and will be used in its derivation.

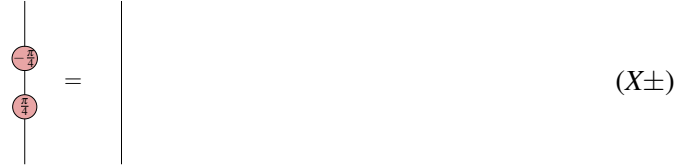
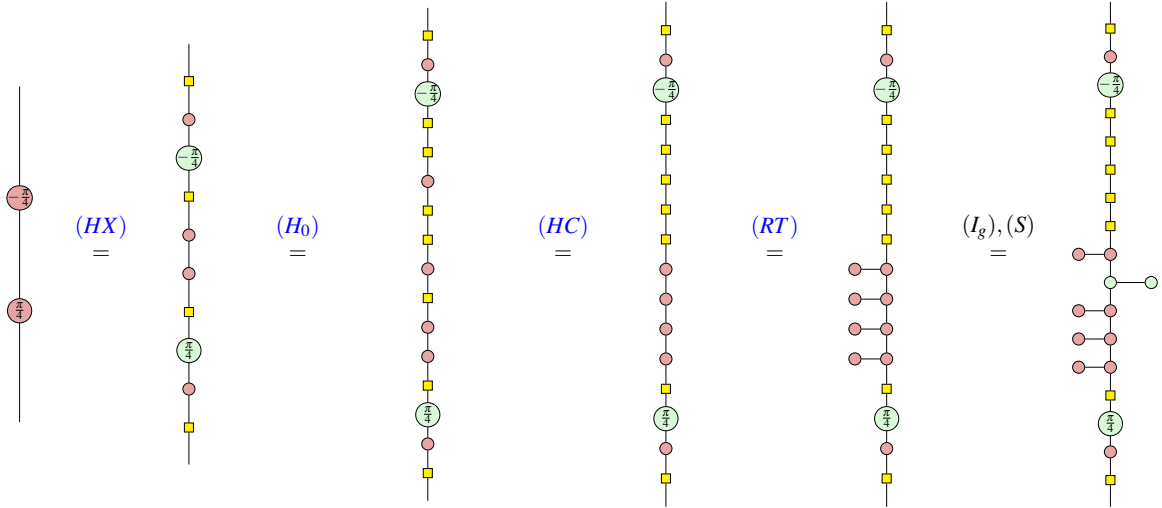
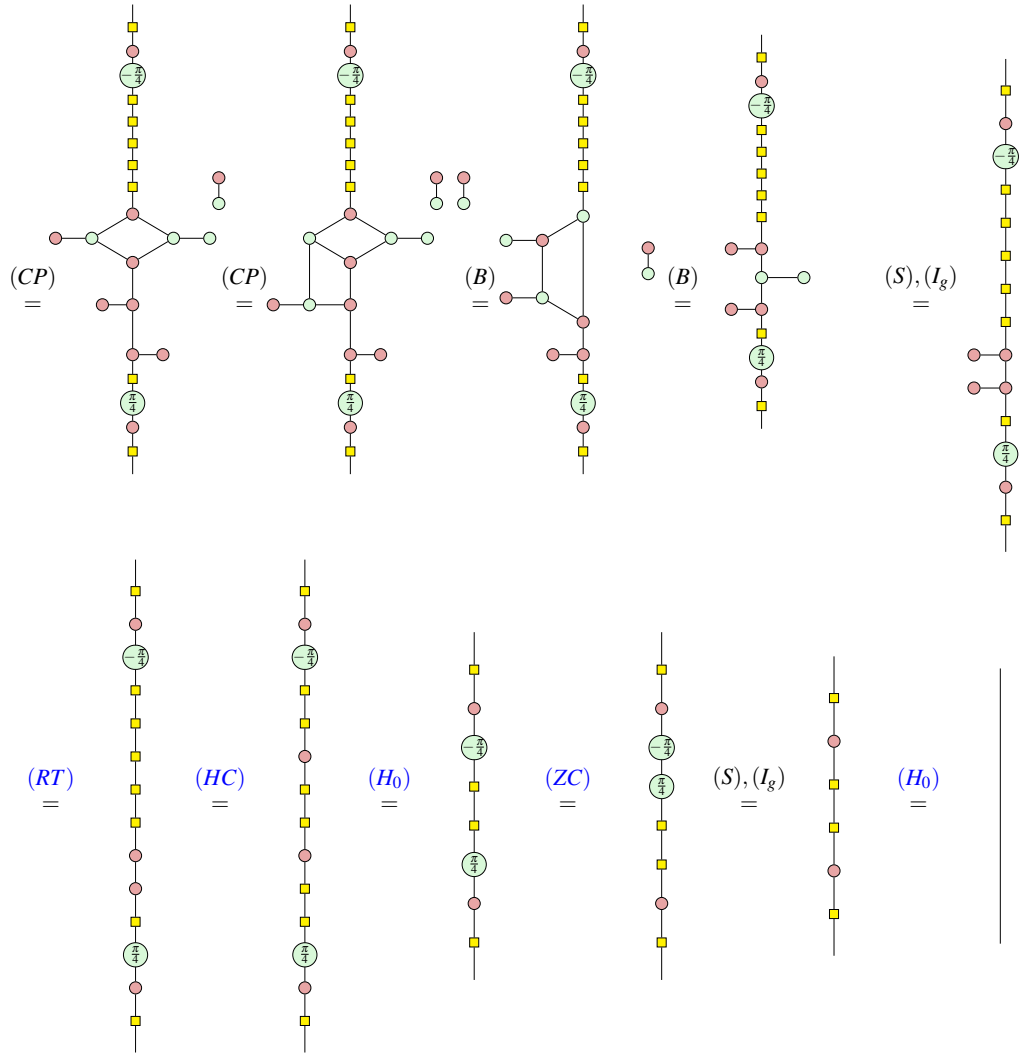
\begin{primeircompactlemma}
\primeirdiagramstatement{quarter-inverse-statement.tikz}{X\!\pm}{QuarterInverse}
\primeircompactproof
\par\noindent\makebox[\linewidth][c]{%
  \primeirtikzinput{xpm-direct-01.tikz}\hspace{1.2cm}%
  \primeirtikzinput{xpm-direct-02.tikz}%
}\par
\par\noindent\makebox[\linewidth][c]{%
  \primeirtikzinput{xpm-direct-03.tikz}\hspace{0.9cm}%
  \primeirtikzinput{xpm-direct-03-tail.tikz}%
}\par
\primeirrow{xpm-direct-04-05.tikz}
\tombstone
\end{primeircompactlemma}

Here we derive the self-inverse property of the Hadamard in our ruleset. Conceptually, this is equivalent to the \hyperref[eq:prime-ir-Ir]{\((I_r)\)} rule, as we will soon see.

\begin{primeircompactlemma}
\primeirdiagramstatement{h2-statement.tikz}{HH}{Htwo}
\primeircompactproof
\par\noindent\makebox[\linewidth][c]{%
  \raisebox{-.5\height}{\primeirtikzinput{h2-general-prefix.tikz}}%
  \hspace{0.15cm}%
  \raisebox{-.5\height}{\primeirtikzinput{h2-convert-first.tikz}}%
}\par
\par\noindent\makebox[\linewidth][c]{%
  \raisebox{-.5\height}{\primeirtikzinput{h2-convert-rest.tikz}}%
  \hspace{0.15cm}%
  \raisebox{-.5\height}{\primeirtikzinput{h2-2pi-01-first.tikz}}%
}\par
\par\noindent\makebox[\linewidth][c]{%
  \raisebox{-.5\height}{\primeirtikzinput{h2-2pi-01-rest.tikz}}%
  \hspace{0.15cm}%
  \raisebox{-.5\height}{\primeirtikzinput{h2-2pi-03.tikz}}%
}\par
\par\noindent\makebox[\linewidth][c]{%
  \raisebox{-.5\height}{\primeirtikzinput{h2-2pi-05.tikz}}%
  \hspace{0.15cm}%
  \raisebox{-.5\height}{\primeirtikzinput{h2-2pi-07.tikz}}%
}\par
\tombstone
\end{primeircompactlemma}

\begin{primeircompacttheorem}[The \((I_r)\) rule is derivable.]
\primeirderivabilitystatement{fig3-ir-theorem-statement.tikz}{I_r}{Ir}
\primeircompactproof
\primeirrow{ir-final-proof.tikz}
\tombstone
\end{primeircompacttheorem}

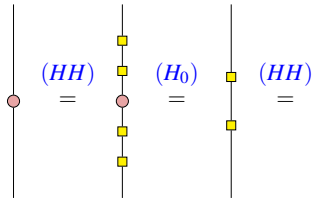
\end{document}